\documentclass[11pt]{article}

\usepackage{fullpage}

\usepackage[utf8]{inputenc}
\usepackage{amsmath}
\usepackage{amsthm}
\usepackage{bbm}
\usepackage{amsfonts}
\usepackage{amssymb}
\usepackage{units}
\usepackage{xcolor}
\usepackage{graphicx}
\usepackage[ruled,titlenotnumbered,linesnumbered,noline,noend]{algorithm2e}
\usepackage{comment}
\usepackage{xfrac}
\usepackage{mathdots}
\let\oldnl\nl
\newcommand{\nonl}{\renewcommand{\nl}{\let\nl\oldnl}}
\newcommand{\etal}{\emph{et al.}}

\usepackage{tikz}
\usetikzlibrary{arrows.meta, positioning, backgrounds}
\usepackage{enumitem}
\tikzset{directed/.style={->, >=Latex, thick}}

\newtheorem{theorem}{Theorem}[section]
\newtheorem{claim}{Claim}[section]
\newtheorem{lemma}{Lemma}[section]

\newcommand{\p}{\mathbb{P}}

\newcommand{\CONGEST}{\textsc{CONGEST}}
\newcommand{\linespace}{\vspace{1.3em}}
\newcommand{\girth}{\operatorname{girth}}
\newcommand{\ex}{\operatorname{ex}}
\newcommand{\z}{\operatorname{z}}
\newcommand{\eps}{\varepsilon}

\newcommand{\fceil}[1]{\lceil #1 \rceil }

\newcommand{\lra}{\leftrightarrows}

\title{Girth Approximations in the CONGEST Model}

\date{}

\author{
Shiri Chechik\thanks{Tel Aviv University, Israel (\texttt{shiri.chechik@gmail.com}, \texttt{gur.lifshitz@gmail.com} and \texttt{doron.muk@gmail.com}).}
\and
Gur Lifshitz\footnotemark[1]
\and
Doron Mukhtar\footnotemark[1]
}

\begin{document}

\maketitle

\begin{abstract} 

This paper advances the state of the art in girth approximation within the CONGEST model. Manoharan and Ramachandran [\textit{PODC ’24}] provided the first significant improvement in girth approximation in over a decade. We build on this momentum and make progress on all fronts: we provide a unified family of algorithms yielding girth approximation–round tradeoffs for undirected networks; we obtain improved bounds for directed networks; and we establish better lower bounds for directed and undirected weighted networks. Together, these results substantially narrow the remaining complexity gaps across all settings.

Specifically, for networks with $n$ nodes and hop-diameter $D$, we show that one can compute, with high probability:

\begin{enumerate}[label=(\arabic*)]
  \item An $f$-approximation for unweighted undirected girth in $\tilde{O}(n^{1/f}+D)$ rounds, for every constant integer $f>2$.
  \item A $(2k-1+o(1))$-approximation for weighted undirected girth in $\tilde{O}(n^{(k+1)/(2k+1)}+D)$ rounds, for every constant integer $k>1$.
  \item A $2$-approximation for directed unweighted girth, and a $(2+\varepsilon)$-approximation for directed weighted girth, both in  $\tilde{O}(n^{2/3}+D)$ rounds.
\end{enumerate}

We also prove new lower bounds for directed networks and for undirected weighted networks: for every integer $k > 2$ and $\varepsilon>0$, assuming the Erd\H{o}s-Simonovits' even cycle conjecture (and unconditionally for $k\in\{3,4,6\}$), any $(k-\varepsilon)$-approximation for the girth requires $\tilde{\Omega}(n^{k/(2k-1)})$ rounds, even when $D = O(\log n)$.


 \bigskip\medskip \end{abstract} 
\section{Introduction} 

Research in distributed computing has primarily focused on designing algorithms for computing important properties of the networks on which they operate. One such property is the girth of the network, the length of its shortest cycle, or more generally, the weight of its lightest cycle (if the network is weighted). Being a fundamental property, there is a natural need to estimate it efficiently. This, however, is far from being a trivial task, especially in models that heavily restrict the communication between the nodes. 

In this paper, we focus on the CONGEST model \cite{Peleg00b} of distributed computing. According to it, we represent a communication network by an $n$-node graph $G$, such that each node corresponds to a device and each edge to a bidirectional communication channel. The network may be weighted and/or directed. In this case, we assume that each node initially knows the weights and/or the directions of its incident edges, and that this information does not affect communication in any way. Computation proceeds in discrete rounds, following a predefined algorithm designed to solve a specific task. In each round, every node receives the messages that were previously sent to it, performs a local computation, and then sends $O(\log{n})$-bit (possibly different) messages to its neighbors. We also assume that each node has the ability to flip coins and use their outcomes during the computation. An algorithm that utilizes this ability is called randomized; otherwise, it is deterministic. A randomized algorithm is said to succeed with high probability if the probability of failure (taken over all coin flips) is at most $1/n^c$ for an arbitrary constant $c \ge 1$.

Our goal is to design an algorithm that enables each node to compute a good approximation of the girth in as few rounds as possible. Prior work has achieved steady progress on this and related problems, yielding both exact and approximation algorithms, which we briefly review.

\subparagraph{Unweighted Networks} For unweighted networks, whether directed or not, we can compute the girth exactly in $O(n)$ rounds using a deterministic all-pairs shortest paths algorithm \cite{HolzerW12, len19, ManoharanR22}. When the edges are directed, this bound is nearly optimal, as $\tilde{\Omega}(n)$ rounds are required for any constant approximation factor below $2$, even when the network's hop-diameter $D$ is constant \cite{ManoharanR22}. In contrast, for undirected networks, the tightness of this result remains unclear, as the best known lower bound for any such approximation factor is $\tilde{\Omega}(\sqrt{n}+D)$ \cite{FritHW12}. 

The landscape of girth approximation in such networks is more complicated. For directed networks, Manoharan and Ramachandran \cite{ManohR24} presented an algorithm that, with high probability, computes a $2$-approximation in $\tilde{O}(n^{4/5} + D)$ rounds, and established that $\tilde{\Omega}(\sqrt{n}+D)$ rounds are necessary for any constant approximation factor at least $2$. For undirected networks, they showed \cite{ManoharanR22, ManohR24} that $(2-1/g)$-approximation (where $g$ is the network's girth) can be computed with high probability in $\tilde{O}(\sqrt{n}+D)$ rounds, while $\tilde{\Omega}(n^{1/(f+1)}+D)$ rounds are required for any constant approximation factor less than $f+0.5$ (for any fixed integer $f > 1$). 

\subparagraph{Weighted Networks} For networks with positive integer weights bounded by a polynomial in $n$, the girth can be computed exactly in $\tilde{O}(n)$ rounds with high probability, regardless of whether the network is directed or not \cite{ManoharanR22, BerN19}. Moreover, \cite{ManohR24} shows that for any constant $\eps > 0$, a $(2+\eps)$-approximation can be obtained with high probability in $\tilde{O}(n^{4/5} + D)$ rounds for directed networks and in $\tilde{O}(n^{2/3} + D)$ rounds for undirected networks. On the lower-bound side, it is known \cite{ManoharanR22, ManohR24} that, independent of edge orientation, any constant-factor approximation better than $2$ requires $\tilde{\Omega}(n)$ rounds, while achieving a factor of $2$ or more requires $\tilde{\Omega}(\sqrt{n}+D)$ rounds. 

\subparagraph{Related Work} A closely related problem is the $q$-cycle detection, which asks to determine whether a given network contains a cycle of length exactly $q$. Detecting undirected cycles of a fixed length $q$ has a complexity of $\tilde{\Theta}(n)$ for odd $q \ge 5$; for even $q \ge 4$, the complexity lies between $\tilde{\Omega}(\sqrt{n})$ and $\tilde{O}(n^{1-2/q})$ \cite{DruckerKO13, KorhonenR17, Frai25}. Triangles can be detected in $\tilde{O}(n^{1/3})$ rounds, even in directed networks \cite{ChangHS24, ManoharanR22}. Detecting longer directed cycles, however, requires $\tilde{\Omega}(n)$ rounds \cite{ManoharanR22}. The girth problem has been also studied in the Congested Clique model, where it can be computed exactly in $O(n^{0.157})$\footnote{The actual running time is $\tilde{O}(n^{1-2/\omega +\eps})$ where $\omega$ is the matrix multiplication exponent and $\eps > 0$ is an arbitrarily small constant.} rounds for unweighted graphs \cite{CenKK19}, and approximated within an additive error of one in a constant number of rounds for undirected, unweighted graphs \cite{CensorG20}.

\linespace Although considerable progress has been made on the girth problem, significant gaps remain between the known upper and lower bounds. In this work, we make an important step toward closing these gaps and obtain the following results:

\subparagraph{Unweighted Undirected Networks} We present new trade-offs between the round complexity and the approximation quality of the girth in unweighted, undirected networks. Specifically, we introduce a new algorithm that, with high probability, computes an $f$-approximation to the girth in $\tilde{O}(fn^{1/f} + fD)$ rounds, for any positive integer $f$. This nearly matches the lower bound of $\tilde{\Omega}(n^{1/(f+1)}+D)$ established in \cite{ManohR24}.

Previously, such trade-offs were believed to be impossible due to a published lower bound, which we have discovered to be flawed: Manoharan and Ramachandran \cite{ManohR24} claimed a stronger lower bound of $\tilde{\Omega}(n^{1/4})$ for any constant approximation factor $\alpha \ge 2.5$, which holds even for networks of diameter $\Theta(\log{n})$. This claim appears as Theorem 2.B in \cite{ManohR24} building on a lower bound for weighted undirected networks established in the proof of Theorem 3.A. Their approach attempts to adapt the construction by replacing each weighted edge with an unweighted path of equal length. However, this transformation overlooks a critical consequence -- it increases the network's diameter from $O(\log n)$ to $\tilde{\Omega}(n^{1/4})$, undermining the claim.

\subparagraph{Weighted Undirected Networks} We obtain new trade-offs for weighted and undirected networks. More precisely, we show that, with high probability, one can compute a $(2k-1+\eps)$-approximation of the girth in $\tilde{O}(kn^{(k+1)/(2k+1)}+kD)$ rounds, for any integer $k > 1$ and constant $\eps > 0$ (assuming positive integral edge weights that are polynomially bounded in $n$).

\subparagraph{Directed Networks} We devise new randomized algorithms that, in $\tilde{O}(n^{2/3} + D)$ rounds, compute a $2$-approximation to the girth in directed unweighted networks, and a $(2 + \eps)$-approximation (for any constant $\eps > 0$) in directed weighted networks (again, with positive integral weights polynomially bounded in $n$). Thus, achieving a polynomial improvement over the previously known upper bounds \cite{ManohR24}.

\subparagraph{Lower Bounds} We complement our new upper bounds for directed networks and for undirected weighted networks with relative lower bounds. For every integer $k > 2$ and $\eps>0$, we show that under the Erd\H{o}s-Simonovits' even cycle conjecture~\cite{ErdosSimonovits1982} (and unconditionally for $k\in\{3,4,6\}$), any $(k-\eps)$-approximation algorithm for the girth in such networks requires $\tilde{\Omega}(n^{k/(2k-1)})$ rounds, even when the hop-diameter $D$ is $O(\log n)$. 

Our result yields a polynomial improvement over the previously known $\tilde{\Omega}(\sqrt{n})$ lower bound of \cite{ManohR24}, and for the first time makes the lower bound explicitly depend on the approximation factor. To illustrate the tightness of our results, observe that for weighted undirected networks our upper and lower bounds essentially match: for example, we obtain a $(3+o(1))$-approximation in $\tilde{O}(n^{3/5}+D)$ rounds, while $(3-o(1))$-approximation requires $\tilde{\Omega}(n^{3/5}+D)$ rounds.

\section{Preliminaries}

All logarithms are base two, and all graphs are simple. For a cycle $C$ in a graph $G$, we denote by $V(C)$ its set of vertices, by $E(C)$ its set of edges and by $|C|$ the size of $E(C)$. If $G$ is weighted, we let $w_G(C)$ denote the weight of $C$ in $G$. For two vertices $s$ and $t$ in $G$, we write $d_G(s,t)$ for the distance from $s$ to $t$ in $G$. Given an additional non-negative integer $h$, we let $d^h_G(s,t)$ denote the $h$-hop distance from $s$ to $t$ in $G$. For a cycle $C$ and a vertex $s$ in $G$, we define the distance from $s$ to $C$ as $d_G(s,C) = \min\{d_G(s,v) \mid v \in V(C)\}$. When the graph $G$ is clear from context, we may omit the subscript $G$ from any of these notations.

Throughout, we assume that each vertex $v$ in an $n$-node communication network is assigned a unique $O(\log{n})$-bit identifier $ID(v)$, initially known only to $v$ itself. If the network is weighted, edge weights are are assumed to lie in $\{1,2,...,c\cdot n^{k}\}$ for some constants $c,k \ge 1$. We also assume that each vertex initially knows the number of vertices $n$, the maximum edge weight $W$ (when applicable) and a constant-factor approximation of the network's hop-diameter $D$. (All these values can be computed deterministically in $O(D)$ rounds using standard techniques; see, e.g., \cite{len19, Peleg00b}.)

\section{Unweighted Undirected Networks} 

In this section, we present new randomized algorithm for approximating the girth of unweighted, undirected networks. 

\subsection{Overview} 

We fundamentally refine the girth approximation framework of \cite{ManoharanR22, ManohR24} by replacing their uniform sampling approach with a multi-scale exploration scheme that systematically traverses the network from a progressively smaller set of sources. This new multi-layered approach, combined with a careful analysis, enables us, for the first time, to obtain an approximation to the girth in much less than $\tilde{\Theta}(\sqrt{n})$ rounds. We now provide a high-level overview of the approach.

As in previous work, we leverage the fact that a partial BFS tree rooted at a vertex $s$ allows one to approximate the girth $g$ whenever the vertices of a shortest cycle $C$ belong to it. Concretely, the girth can be estimated as the minimum of $d(s,x)+d(s,y)+1$ over all crossing edges $\{x,y\}$, yielding an additive approximation with error at most $2d(s,C)$. 

Intuitively, applying this procedure from several sources and taking the minimum over all estimates can lead to a smaller additive error, as we may hit a closer source. Manoharan and Ramachandran \cite{ManoharanR22, ManohR24} exploit this idea as follows. First, they construct partial BFS trees from all vertices so that each one participates only in the trees of its $O(\sqrt{n})$ closest vertices, and compute an initial estimate. (This can be done in $O(\sqrt{n}+D)$ rounds by using a source-detection algorithm \cite{len19}.) If the vertices of a shortest cycle $C$ are fully contained in the partial BFS tree rooted at some vertex of $C$, then the girth is computed exactly. Otherwise, there must be $\Omega(\sqrt{n})$ vertices whose disance from $C$ is at most $g/2$. Consequently, by sampling each vertex independently with probability $\tilde{\Theta}(n^{-1/2})$, one of them is selected with high probability. Thus, construting full BFS trees from the sampled sources and estimating again yields a $2$-approximation to the girth. This takes another $\tilde{O}(\sqrt{n} + D)$ rounds, as there are $\tilde{O}(\sqrt{n})$ sampled sources with high probability.

We further extend this idea and show that one can trade approximation quality for running time by performing several such estimations at different sampling scales. Specifically, we perform $f$ estimations, where in the $i$-th estimation vertices are sampled independently with probability $\tilde{\Theta}(n^{-i/f})$, and partial BFS trees are constructed from the sampled sources such that each vertex participates only in the trees of its $\tilde{O}(n^{1/f})$ closest sampled sources. (This takes $\tilde{O}(n^{1/f} + D)$ rounds.) By generalizing the previous argument, we can show that for some $i \in \{0,...,f-1\}$, there exists a sampled vertex $x$ in this layer whose distance to a shortest cycle $C$ is at most $(g/2) \cdot i$, and such that $x$ belongs to the set $Q_{i+1}(v)$ of the $O(n^{(i+1)/f})$ closest vertices of each vertex $v$ in $C$. Since each vertex is chosen independently with probability $p_i = \tilde{\Theta}(n^{-i/f})$, it follows with high probability that $Q_{i+1}(v)$ contains $p_i \cdot |Q_{i+1}(v)| = \tilde{O}(n^{1/f})$ sampled vertices for every $v$ in $C$. Consequently, all vertices of $C$ participate in the partial BFS tree rooted at $x$, yielding an $(i+1) \le f$ approximation in $\tilde{O}(f \cdot n^{1/f} + f \cdot D)$ rounds.

The rest of this section is organized as follows. We first formalize and generalize the estimation procedure that appears implicitly in earlier work. We then show how to incorporate this procedure into a complete girth-approximation algorithm, thereby achieving the claimed tradeoff between approximation quality and running time.

\subsection{The Estimation Procedure} 

Let $G = (V,E)$ be an undirected network with girth $g$ and diameter $D$. Let $S \subseteq V$ be a set of sources and $k$ be a natural number. We say that a vertex $v \in V$ is closer to $u \in V$ than $w \in V$, if either $d(u,w) > d(u,v)$, or $d(u,w) = d(u,v)$ and $ID(w) > ID(v)$. For each vertex $v \in V$, we let $U(S,k,v)$ be the set of the $k$ closest vertices to $v$ from $S$, or simply $S$ if $|S| \le k$. We say that a cycle $C$ in $G$ is covered by $S$, if there is a source $s \in S$ that covers it, i.e., $s \in U(S,k,v)$ for all $v \in V(C)$. In this case, we denote by $d(S,C)$ its distance from a closest source that covers it, that is, $d(S,C) = \min\{d(s,C) \mid s \in S \text{ is a source that covers } C\}$. Let $\mathcal{C}$ be the set of all the shortest cycles in $G$ that $S$ covers. In the case where $\mathcal{C}$ is nonempty, we define $d(S)$ to be the smallest distance in $G$ between a shortest cycle and a source that covers it, i.e., $d(S) = \min\{d(S,C) \mid C\in\mathcal{C}\}$.

Our goal in this section is to describe a deterministic algorithm that computes in $O(k+D)$ rounds a value $M$ known to every vertex, such that $M \ge g$ always holds, and $2d(S) + g \ge M$ holds when $\mathcal{C}$ is not empty. To achieve this, we use the source detection algorithm \cite[Theorem 1.2]{len19}, which allows us to do the following:

\begin{lemma}\label{lem:src-dec} For an unweighted and undirected network $G = (V,E)$ with diameter $D$, a set of sources $S \subseteq V$ and a natural number $k$, partial BFS trees rooted at the sources can be deterministically constructed in $O(k+D)$ rounds, such that each vertex participates only in the trees of its $\min(k,|S|)$ closest sources. At the end, each vertex $v$ knows the set $U(S,k,v)$, along with the corresponding distances $d(s,v)$ and parent pointers $p(s,v)$, for all $s \in U(S,k,v)$.\end{lemma}

We now describe the procedure for computing $M$, which internally uses the source detection algorithm of Lemma \ref{lem:src-dec}.

\linespace

\begin{procedure}[H] \caption{Estimate($G, S, k$)}
\DontPrintSemicolon
 \smallskip
	\SetKwInOut{Input}{Input} \Input{An undirected and unweighted network $G = (V,E)$, a set of sources $S \subseteq V$ and a natural number $k$, where each vertex initially knows $k$ and whether it belongs to $S$.} 
	
	\smallskip
	The vertices run the source detection algorithm on $(G,S,k)$.\\Each vertex $v$ broadcasts the set $\{(s, d(s,v), p(s,v)) \mid s \in U(S,k,v)\}$ to its neighbors. \label{line:b}
	\\ Every vertex $y$ performs the following internal computation:\\ \Indp$M_y \gets \infty$\\
 \ForEach{\upshape neighbor $x$ \textbf{and} source $s \in U(S,k,y) \cap U(S,k,x)$}{
\lIf{\upshape $p(s,x) \neq y$ \textbf{and} $p(s,y) \neq x$}{$M_y \gets \min(M_y, d(s,x) + d(s,y) + 1)$\,} \label{eq:1}
 }\Indm
 
 The vertices compute $M = \min\{M_v \mid v\in V\}$ by convergecasting (see, e.g., \cite{Peleg00b}).

\end{procedure}

\linespace

The following lemma establishes the correctness of the procedure by showing that the computed value $M$ satisfies the required bounds.

\begin{lemma} For an undirected and unweighted network $G = (V,E)$ with girth $g$, a set of sources $S \subseteq V$ and a natural number $k$, procedure $\text{Estimate}(G,S,k$) computes a globally known value $M$ satisfying $g \le M$, and $M \le 2d(S) + g$ whenever $\mathcal{C}$ is nonempty.\end{lemma}

\begin{proof} We start by showing that $M \ge g$. For this we have to show that $d(s,x) + d(s,y) + 1 \ge g$ holds for every edge $\{x,y\} = e \in E$ and every source $s \in U(S,k,x) \cap U(S,k,y)$ that satisfy the condition at line \ref{eq:1}. Fix such $e = \{x,y\}$ and $s$. As $p(s,x) \not= y$, there is a path $p$ of length $d(s,x)$ from $s$ to $x$ such that $y$ is not the parent of $x$ in $p$. We claim that $y$ does not appear in $p$ at all. Indeed, if it was, then it would be at distance at least $2$ from $x$ in $p$, but then we could replace this subpath by the edge $\{x,y\}$ and get a shorter path, in contradiction. By a symmetric argument, there is a path $q$ of length $d(s,y)$ from $s$ to $y$ such that $x$ does not appear in it. Let $u$ be the closest vertex to $y$ in $q$ that also appears in $p$ (such a vertex must exists as $s$ is a common vertex). Let $p'$ and $q'$ be the subpaths of $p$ and $q$ from $x$ to $u$ and from $u$ to $y$, respectively. It is easy to see that $C = p' \circ q' \circ (y,x)$ is a cycle of length at most $d(s,x) + d(s,y) + 1$, and so $d(s,x) + d(s,y) + 1 \ge |C| \ge g$.

Assume now that $\mathcal{C}$ is not empty, and let $C \in \mathcal{C}$ be such that $d(S) = d(S,C)$. Let $s \in S$ be a source that covers $C$ such that $d(S,C) = d(s,C)$. Since $s$ covers $C$, we have $s \in U(S,k,v) \cap U(S,k,u)$ for every edge $\{u,v\}$ in $C$. Moreover, at least one of these edges $e = \{x,y\}$ must satisfy the condition at line \ref{eq:1} with respect to $s$, as $C$ is a cycle. Thus, for $s$ and $x$, we have $M_{x} \le d(s,x) + d(s,y) + 1$. Let $v$ be a vertex in $V(C)$ such that $d(s,C) = d(s,v)$. Note that $d(S) = d(S,C) = d(s,C) = d(s,v)$. By the triangle inequality we know that $d(s,x) \le d(s,v) + d(v,x)$ and $d(s,y) \le d(s,v) + d(v,y)$. It follows that $M_{x} \le 2d(s,v) + d(v,x) + d(v,y) + 1$. As $d(v,x) + d(v,y) + 1$ cannot be more than the length of $C$, we get that $M \le M_{x} \le 2d(s,v) + |C| = 2d(S) + g$.\end{proof} 

The next lemma establishes an upper bound on the round complexity of the procedure as a function of $k$ and $D$.

\begin{lemma} Procedure $\text{Estimate}(G,S,k$) runs in $O(k+D)$ rounds on any undirected, unweighted network $G$ of diameter $D$.\end{lemma}

\begin{proof} The source detection algorithm takes $O(k+D)$ rounds (Lemma \ref{lem:src-dec}), broadcasting the set in line \ref{line:b} takes $O(k)$ rounds, and computing $M$ takes $O(D)$ rounds. Thus, $O(k+D)$ rounds in total.\end{proof} 

\subsection{The Approximation Algorithm} 

Let $f$ be a positive integer, and $c \ge 1$ a constant. We describe a randomized algorithm that, with probability at least $1 - n^{-c}$, computes in $\tilde{O}(f \cdot n^{1/f} + f \cdot D)$ rounds, an $f$-approximation to the girth of any undirected network $G = (V,E)$ with $n > 0$ vertices and diameter $D$.

\linespace

\begin{algorithm}[H] \caption{Approximating the girth in unweighted undirected networks}
\DontPrintSemicolon
\tcp{Let $A_0$ be $V$ and $k$ be $\fceil{12(c+3)n^{1/f}\log{n}}$.}
\eIf{\upshape $n \leq f$ }{
    $M \gets \text{Estimate}(G, A_0, n)$.
}{
    $M_0 \gets \text{Estimate}(G, A_0, k)$.\\
    \For{$i \gets 1$ \KwTo $f-1$}{
       Each vertex samples itself with probability $p_i = \min((c+3)n^{-i/f}\log{n},1)$.\\    
      $M_i \gets \text{Estimate}(G, A_i, k)$, where $A_i$ is the sampled set.\\
    }
    Each vertex locally computes $M = \min(M_0,...,M_{f-1})$.
}

\end{algorithm}\linespace

Next, we establish the correctness of the algorithm and provide a thorough analysis of its running time.

\subparagraph{Correctness.} If $n \le f$ then the algorithm computes the girth exactly (as, in this case, any vertex in a shortest cycle covers it). So assume that $n > f$. For each $i \in \{1,...,f\}$ and each vertex $v \in V$, let $Q_i(v)$ be the set of the $\fceil{n^{i/f}}$ closest vertices to $v$ in $G$\footnote{Recall that a vertex $x \in V$ is said to be closer to $v$ than a vertex $y \in V$, if either $d(v,x) < d(v,y)$, or $d(v,x) = d(v,y)$ and $ID(x) < ID(y)$.}. We say that the sampling is good, if the following two events occur: (1) for every $i \in \{1,...,f-1\}$ and every $v \in V$, the set $Q_i(v)$ contains a vertex from $A_i$, and (2) for every $i \in \{0,...,f-1\}$ and every $v \in V$, the set $Q_{i+1}(v)$ contains at most $\fceil{12(c+3)n^{1/f}\log{n}}$ vertices from $A_i$. We start by proving that the sampling is good with high probability.
 
\begin{claim} The sampling is good with probability at least $1 - n^{-c}$. \label{claim:key_pro}\end{claim}

\begin{proof} By the union bound, it is enough to show that each of the events fails to occur with probability at most $n^{-c-1}$. (1) For a given $i \in \{1,...,f-1\}$ and $v \in V$, the probability that $Q_i(v)$ does not contain any vertex from $A_i$ is $$ (1-p_i)^{|Q_i(v)|} =\begin{cases} 0, & \text{if $p_i = 1$}\\ (1-(c+3)n^{-i/f}\log{n})^{\fceil{n^{i/f}}} \le e^{-(c+3)\log{n}} \le n^{-c-3}, & \text{otherwise} \end{cases}$$ thus, by the union bound over all the vertices and all the indices, we get that the event fails to occur with probability at most $n^{-c-1}$. (2) Let $i \in \{0,...,f-1\}$ and $v \in V$. If $i = 0$, then $Q_1(v)$ can contain at most $|Q_1(v)| = \fceil{n^{1/f}} \le \fceil{12(c+3)n^{1/f}\log{n}}$ vertices from $A_0$, so the probability that the event does not occur is $0$. Otherwise, the expected number of vertices $X_v^i$ in $Q_{i+1}(v)$ that belong to $A_i$ is $$|Q_{i+1}(v)| \cdot p_i \le \fceil{n^{(i+1)/f}} \cdot (c+3)n^{-i/f}\log{n} \le 2(c+3)n^{1/f}\log{n}$$ thus, by Chernoff bound, we have $$P(X_v^i \ge 12(c+3)n^{1/f}\log{n}) \le 2^{-12(c+3)n^{1/f}\log{n}} \le n^{-c-3}$$ and so by the union bound over all the vertices and indices, we get that the event fails to occur with probability at most $n^{-c-1}$.\end{proof}

Now, we will show that when the sampling is good, the algorithm finds an $f$-approximation. Let $g$ be the girth of $G$. If it is not finite, then all the approximations will return $\infty$. Otherwise, $G$ has a shortest cycle $C$. We prove the following key property:

\begin{claim} There exist an index $i \in \{0,...,f-1\}$ and a vertex $x \in A_i$ such that $d(x,C) \le (g/2)\cdot i$, and $x \in Q_{i+1}(u)$ for all $u \in V(C)$.\label{claim:main_obs}\end{claim}

\begin{proof} For the sake of contradiction, suppose that the claim is false. We will proceed by induction on $i$ to show that for each $i \in \{0,...,f-1\}$ there exists a vertex $x \in A_{i}$ such that $d(x,C) \le (g/2)\cdot i$. This will imply that for $i = f-1$ there is a vertex $x \in A_{f-1}$ such that $d(x,C) \le (g/2)\cdot (f-1)$, and so that there is a vertex $u \in V(C)$ such that $x \not\in Q_{f}(u)$, which is impossible, as $Q_{f}(u) = V$ by definition.

The base case ($i = 0$) is trivial as for any vertex $x \in V(C) \subseteq V = A_0$, we have $d(x,C) = 0 \le (g/2)\cdot i$. Assume that the claim is true for $0 \le i < f-1$, and prove it for $i+1$. By the induction hypothesis, there exists a vertex $y \in A_i$ such that $d(y,C) \le (g/2)\cdot i$. As we assumed that the claim does not hold, there must be a vertex $u \in V(C)$ such that $y \not\in Q_{i+1}(u)$. As the sampling is good, there is a vertex $x \in A_{i+1} \cap Q_{i+1}(u)$. As $y \not\in Q_{i+1}(u)$ and $x \in Q_{i+1}(u)$, we have $d(u,x) \le d(u,y)$. Let $z \in V(C)$ be such that $d(y,z) = d(y,C)$. By the triangle inequality, $d(u,y) \le d(u,z) + d(z,y)$, and so \begin{align*} d(x,C) \le d(x,u) &\le d(u,y) \le d(u,z) + d(z,y) = d(u,z) + d(y,C) \\ &\le g/2 + gi/2 = (g/2)\cdot(i+1) \qedhere\end{align*}\end{proof}

In the next claim, we use this property to show that the algorithm returns a good approximation:

\begin{claim} There exists $i \in \{0,...,f-1\}$ such that $M_i \le g(i+1)$\label{claim:good_app}\end{claim}

\begin{proof} By Claim \ref{claim:main_obs}, there exist $i \in \{0,...,f-1\}$ and $x \in A_i$ such that $d(x,C) \le gi/2$, and $x \in Q_{i+1}(u)$ for all $u \in V(C)$. As the sampling is good, we know that for each $u \in V(C)$, the set $Q_{i+1}(u)$ contains at most $k = \fceil{12(c+3)n^{1/f}\log{n}}$ vertices from $A_i$. This implies that $Q_{i+1}(u) \cap A_i \subseteq U(A_i,k,u)$ for all $u \in V(C)$. In particular, we have $x \in U(A_i,k,u)$ for all $u \in V(C)$. It follows that $x$ covers $C$, and so $M_i \le 2d(A_i) + g \le 2d(x,C) + g \le g(i+1)$.\end{proof}

We conclude that the value $M = \min(M_0,...,M_{f-1})$ that the algorithm returns is always in the range $[g,f \cdot g]$.

\subparagraph{Complexity.} When $n \le f$, we compute one approximation that takes $O(f+D)$ rounds. Otherwise, we compute $f$ approximations, each takes $O(n^{1/f}\log{n} \cdot c + D)$ rounds. Thus, the complexity is $\tilde{O}(f \cdot n^{1/f} \cdot c + f \cdot D)$.
\section{Weighted Undirected Networks} 

In this section, we present a randomized algorithm that achieves new trade-offs for approximating the girth of weighted, undirected networks.

\subsection{Overview}

A primary bottleneck in the framework of Manoharan and Ramachandran \cite{ManohR24} is the $\Theta(|S|^2)$ rounds required to compute approximate distances from a set of sources $S$ (used to estimate the weight of certain cycles), via a procedure that relies on an exhaustive broadcast of all-pairs distances within $S$ as an intermediate step. We overcome this limitation by demonstrating that such dense dissemination is not an inherent requirement. By substituting the full broadcast with a strategically constructed spanner, we effectively reduce the communication cost. This technique breaks the quadratic dependency on $|S|$, effectively accelerating the girth estimation process at the cost of a marginal decrease in approximation quality, which remains close to optimal. We now provide a more detailed overview of this approach.

The first observation of \cite{ManohR24} is that the approximation algorithm for unweighted networks can be adapted to weighted ones by simulating each edge of weight $w$ as an unweighted path of length $w$. This transformation allows one to compute a $2$-approximation to the girth in weighted undirected networks, but may require as many as $\Omega(WD)$ rounds, where $W$ is the maximum edge weight and $D$ is the hop-diameter. To mitigate this, they limit the depth of the BFS scans to at most $h$ hops, which yields the following result \cite[Corollary 12]{ManohR24}:

\begin{lemma}\label{lem:bhopapp} Given a weighted undirected network $G$ and a natural number $h$, we can, with high probability, compute in $\tilde{O}(\sqrt{n} + h + D)$ rounds an approximation $M$ to the weighted girth $g$ such that $g \le M$ and $M \le 2g$ whenever $g \le h$.\end{lemma}

Next, they used a scaling technique \cite{Nan14} to move from an approximation for bounded weight to one for bounded hop. Specifically they construct $O(\log{hW})$ versions of $G$, where in the $i$-th version $G_i$, the weight of each edge $e$ is set to $\lceil 2hw(e)/\epsilon 2^i\rceil$. By an analysis similar to that of \cite[Theorem 3.3]{Nan14}, we can show that if $G$ contains a minimum weight cycle $C$ with at most $h$ edges, then there exists an index $i$ such that $w_{G_i}(C) = O(h)$, and $(\epsilon2^i/2h) \cdot w_{G_i}(C) \le (1+\epsilon)w(C)$. Thus, invoking Lemma \ref{lem:bhopapp} on $G_i$ gives an approximation $M_i$ such that $M_i \le 2w_{G_i}(C)$, and so $(\epsilon2^i/2h) \cdot M_i \le (2+2\epsilon)w(C)$. Applying this procedure to all graphs and taking the minimum over the resulting values yields a $2(1+\epsilon)$-approximation for the minimum weight cycle with hop bound $h$ in $\tilde{O}(\sqrt{n} + h + D)$ rounds.

This approach gives a good running time when $h$ is relatively small. To obtain an approximation for larger values of $h$, they use a direct sampling technique aimed at hitting a minimum-weight cycle. Specifically, if $G$ contains a cycle $C$ with at least $h$ vertices, then selecting each vertex independently with probability $\tilde{\Theta}(1/h)$ ensures that, with high probability, at least one vertex of $C$ is selected. Thus, computing $(1+\epsilon)$-aproximate distances from the selected set $S$, and computing an estimation for all non-tree's edges, gives us a $(1+\epsilon)$-approximation to the minimum-weight cycle. Computing these distances takes $\tilde{O}(h+D+|S|^2)$ rounds, e.g. \cite[Algorithm~1]{ManohR24}. Since $|S| = \tilde{O}(n/h)$ with high probability, the total running time is $\tilde{O}(h + (n/h)^2+D)$. Setting $h = \Theta(n^{2/3})$ yields a running time of $\tilde{O}(n^{2/3}+D)$ for approximating cycles with $\Omega(n^{2/3})$ edges. The rest can be approximated with the previous procedure in another $\tilde{O}(n^{2/3} + D)$ rounds.

The critical overhead of this approach lies in the approximation procedure for cycles with $\Omega(n^{2/3})$ hops, as it requires to compute approximate distances from $|S| = \tilde{\Theta}(n/h)$ sources using a procedure that takes $\Omega(h+D+|S|^2)$ rounds. This procedure operates by first computing approximate $h$-hop distances from all sources, then constructing an overlay network on the set $S$, and finally broadcasting all edges of this network, enabling each vertex to approximate its distance from each source. Here, we show that broadcasting the entire network (which may take $\Omega(|S|^2)$ rounds) can be avoided by constructing an appropriate spanner and broadcasting only it, thereby trading approximation quality for improved running time.

The remainder of this section is organized as follows. We first describe our new distance approximation procedure, which relies on spanner construction. Next, we integrate this procedure into the girth-approximation algorithm and reanalyze it.

\subsection{Multi-Source Distance Approximation} 

Let $G = (V,E,w)$ be a weighted undirected network with $n > 0$ vertices and hop-diameter $D$, and let $h$ and $k$ be two natural numbers. The goal of this section is to describe an axuliary procedure that given a sampled set of sources $S \subseteq V$, obtained by choosing each vertex with probability $\tilde{\Theta}(1/h)$, computes a $(2k-1)(1+\epsilon)$-approximation to the distances from all $S$. We begin by briefly describing two tools that we will use in our algorithm.

\subparagraph{Approximating Bounded Hop Distances.} We rely on the following known result on multi-source $h$-hop distance
approximation \cite[Theorem 3.6]{Nan14}, which we summerize in the following lemma:

\begin{lemma} \label{lem:bhop} There exists a randomized algorithm that, given a weighted undirected network $G = (V,E,w)$ with hop-diameter $D$, a set of sources $S \subseteq V$, an integer hop parameter $h \ge 1$, and an approximation parameter $\epsilon > 0$, computes distance estimates $\tilde{d}(s,v)$ for all $(s,v) \in S \times V$, such that with high probability
\[
d(s,v) \le \tilde{d}(s,v) \le (1+\epsilon)\, d^{h}(s,v)
\]
where $\tilde{d}(s,v)$ is known to $v$. The algorithm runs in $\tilde{O}(h/\epsilon + |S| + D)$ rounds. \end{lemma}

\subparagraph{Construction of Spanners on Overlay Networks.} We say that a weighted undirected graph $H = (V_H, E_H, w_H)$ is an \emph{overlay network} (or equivalently, a \emph{virtual network}) on a communication network $G = (V_G, E_G, w_G)$ if $V_H \subseteq V_G$ and each vertex $v \in V_G$ knows whether it belongs to $V_H$, and if it does, it also knows the set of edges incident to it in $H$. 

It is known that for every weighted undirected network on $n$ vertices and every natural number $k$, a $(2k-1)$-spanner of size $O(kn^{1+1/k})$ can be deterministically computed in $\tilde{O}(kn^{1/k})$ rounds in the Broadcast Congested Clique \cite[Corollary 5.1]{Bek21}. This algorithm consists of $O(k)$ iterations in each of which every vertex broadcasts $\tilde{O}(n^{1/k})$ bits, resulting in a total communication of $\tilde{O}(n^{1+1/k})$ bits per iteration. For an overlay network $H$ with $n$ vertices over a real network with hop-diameter $D$, each such iteration can be simulated in $\tilde{O}(n^{1+1/k} + D)$ rounds (by standard casting operations, e.g. \cite{Peleg00b}), which gives the following result:

\begin{lemma} \label{lem:spanner} For any natural number $k$ and an overlay network $H = (V_H, E_H, w_H)$ over $G$ with hop-diameter $D$, a $(2k-1)$-spanner of size $O(k|V_H|^{1+1/k})$ can be deterministically computed in $\tilde{O}(k|V_H|^{1+1/k} + kD)$ rounds. \end{lemma}

Next, we describe the approximation algorithm, provide a proof of its correctness, and analyze its running time.

\linespace

\begin{procedure}[H] \caption{DisApprox($G, S, h, k, \epsilon$)}
\DontPrintSemicolon
 \smallskip
	\SetKwInOut{Input}{Input} \Input{An undirected weighted network $G = (V,E,w_G)$, a set of sampled sources $S \subseteq V$ (obtained by choosing each vertex with probability $\min((c+2)\log{n}/h,1)$, for some natural number $h$), a natural number $k$ and an approximation parameter $\epsilon > 0$. Each vertex initially knows $k, \epsilon$ and whether it belongs to $S$.} 
	
	\smallskip
	
	The vertices compute $(2h)$-hop distance approximation $\tilde{d}$ from $S$ by using Lemma \ref{lem:bhop}.\\\label{step:oneH} \smallskip
	\tcc{Let $H$ be the overlay network whose vertex set is $S$, with an edge between every pair of sources that have a finite estimated distance, and where the weight of each edge equals the corresponding distance estimate $\tilde{d}$. Note that $H$ is implicitly constructed in step~\ref{step:oneH}.} \smallskip
	
	The vertices construct a $(2k-1)$-spanner $H'$ for $H$ using Lemma~\ref{lem:spanner}, and broadcast it.\\\label{step:twoS} 
	Each vertex $v$ locally computes:\\ \Indp
 \ForEach{\upshape $s \in S$}{
 $d_{out}(s,v) \gets \min\{d_{H'}(s,s') + \tilde{d}(s',v)\mid s' \in S\}$

 }\Indm

\end{procedure}\vspace{0.6em}

\begin{lemma} With high probability, the computed distances approximate the true distances within a factor of $(2k-1)(1 + \epsilon)$. \end{lemma}

\begin{proof} The proof follows a fairly standard argument. For each vertex pair $(u,v)$ with a shortest path of $h$ vertices in $G$, choose one such path $\pi(u,v)$, and let $\Pi$ denote the set of all the chosen paths. We say that $S$ hits $\Pi$ if every path in $\Pi$ contains at least one vertex from $S$. Observe that $S$ hits $\Pi$ with high probability: for any path $\pi \in \Pi$, the probability that it contains no vertex from $S$ is at most $e^{-(c+2)\log{n}} \le n^{-(c+2)}$, and as $|\Pi| \le n^2$, it follows by the union bound that $S$ fails to hit $\Pi$ with probability at most $n^{-c}$. The rest of the proof establishes that, whenever $S$ hits $\Pi$, the procedure yields the correct approximation.

Consider a pair $(s,v) \in S\times V$, and let $p$ be a shortest path from $s$ to $v$ in $G$. We partition $p$ into $m$ consecutive subpaths $p_1,...,p_m$, each containing $h$ vertices, except possibly the last, which contains at most $h$ vertices. Without loss of generality, we may assume that each path in $\{p_2,...,p_{m-1}\}$ belongs to $\Pi$, and so each such $p_i$ contains a vertex $s_i \in S$. Let $(x_1,...,x_t)$ denote $(s, s_2,...,s_{m-1}, v)$, and note that for each two consecutive vertices in it, there exists a shortest path containing at most $2h$ edges. Thus, with high probability, we have $\tilde{d}(x_i,x_{i+1}) \le (1+\epsilon)d^{2h}_G(x_i,x_{i+1}) = (1+\epsilon)d_G(x_i,x_{i+1})$ for all $i \in \{1,...,t-1\}$. It follows that $d_H(x_1,x_{t-1}) \le (1+\epsilon)\sum_{i=1}^{t-2}d_G(x_i,x_{i+1})$, and so $d_{H'}(x_1,x_{t-1}) \le (2k-1)d_{H}(x_1,x_{t-1}) \le (2k-1)(1+\epsilon)\sum_{i=1}^{t-2}d_G(x_i,x_{i+1})$. We conclude that $d_{out}(s,v) \le d_{H'}(x_1,x_{t-1}) + \tilde{d}(x_{t-1},x_t) \le (2k-1)(1+\epsilon)\sum_{i=1}^{t-1}d_G(x_i,x_{i+1}) = (2k-1)(1+\epsilon)d_G(x_1,x_t)$.\end{proof}

\begin{lemma}\label{lem:runEsti} The procedure takes $\tilde{O}(h/\epsilon + k|S|^{1+1/k}+kD)$ rounds for any network whose hop-diameter is $D$.\end{lemma}

\begin{proof} Step \ref{step:oneH} takes $\tilde{O}(h/\epsilon + |S| + D)$ rounds, and step \ref{step:twoS} takes $\tilde{O}(k|S|^{1+1/k}+kD)$ rounds. Thus, $\tilde{O}(h/\epsilon + k|S|^{1+1/k}+kD)$ rounds, in total.\end{proof}

\subsection{The Approximation Algorithm} 

Let $k > 1$ be an integer, and $\epsilon > 0$ be a constant. We describe a randomized algorithm that, with high probability, computes in $\tilde{O}(kn^{(k+1)/(2k+1)}+kD)$ rounds, a $(2k-1+\epsilon)$-approximation of the girth for any undirected weighted network $G$ with $n > 0$ vertices and hop-diameter $D$. 
\linespace

\begin{algorithm}[H] \caption{Approximating the girth in weighted undirected networks}\label{alg:weiundire}
\DontPrintSemicolon

	\SetKwInOut{Input}{Input} \Input{An undirected weighted network $G = (V,E,w)$, a natural number $k > 1$ and an approximation parameter $\epsilon > 0$. Each vertex initially knows $k$ and $\epsilon$.} \medskip
Each vertex sets $h \gets \lceil n^{\frac{k+1}{2k+1}} \rceil$ and samples itself with probability $\min((c+2)\log{n}/h,1)$.\\\medskip
\tcp{Compute an estimation for long cycles ($S$ is the set of sampled vertices)}

The vertices call DisApprox($G,S,h,k,\epsilon$) to compute distance estimations $d'$, where during its execution, each vertex $v$ also keeps track of its parent $p(s,v)$ along the estimated path from each $s \in S$.\\ \label{line:invoke}

Each vertex $v$ broadcasts the set $\{(d'(s,v),p(s,v)) \mid s \in S\}$ to its neighbors.\\ \label{line:broa}
Every vertex $v$ locally computes:\\ \Indp$M_v \gets \infty$\\
 \ForEach{\upshape neighbor $x$ \textbf{and} source $s \in S$}{
\lIf{\upshape $p(s,x) \neq v$ \textbf{and} $p(s,v) \neq x$}{$M_v \gets \min(M_v, d'(s,v) + d'(s,x) + w(v,x))$\,} \label{line:condweight}
 }\Indm
 
 The vertices compute $M_l = \min\{M_v \mid v\in V\}$ by convergecasting \cite{Peleg00b}. \label{line:broaMs}
 
 \medskip
\tcp{Compute an estimation for short cycles}
Each vertex sets $h^{*} \gets (1+2/\epsilon) \cdot h$.\\
\For{\upshape $i \gets 1$ \KwTo $\lceil \log{hW}\rceil $}{ \label{line:loop}
Each vertex locally adjusts the weights of its incident edges so that each edge $e \in E$ has weight $\lceil 2hw(e)/\epsilon 2^i\rceil$. \tcp{Denote by $G_i$ the resulting graph}
The vertices compute an $h^*$-limited $2$-approximation of $G_i$'s girth $M_i$ using Lemma \ref{lem:bhopapp}. }

Each vertex locally computes $M_s = \min_i\{\epsilon2^iM_i/2h\}$. \medskip

\tcp{Compute the final estimation}
Each vertex locally computes $M = \min(M_s,M_l)$.
\end{algorithm}

\linespace
We prove that the algorithm correctly computes the claimed approximation and analyze its running time.

\subparagraph{Correctness.} 
Let $C$ be a minimum weight cycle in $G$. We first show that $w(C) \le M_s$ and, if $C$ contains $<h$ edges then $M_s \le (2+\epsilon) w(C)$. For each $i \in \{1,...,\lceil \log{hW}\rceil\}$, the weight of any cycle $C'$ in $G_i$ is \[\sum_{e \in E(C')}\lceil 2hw(e)/\epsilon 2^i\rceil \ge (2h/\epsilon 2^i)\sum_{e \in E(C')}w(e) = (2h/\epsilon2^i)w(C') \ge (2h/\epsilon 2^i)w(C)\] It follows that $M_i \ge (2h/\epsilon2^i)w(C)$, and so $\epsilon2^iM_i/2h \ge w(C)$. Thus, $M_s \ge w(C)$. Assume now that $C$ contains less than $h$ edges. Let $i \in \{1,...,\lceil \log{hW}\rceil\}$ be such that $2^{i-1} \le w(C) < 2^i$. The weight of $C$ in $G_i$ is $w_{G_i}(C)=\sum_{e \in E(C)}\lceil 2hw(e)/\epsilon2^i\rceil < \sum_{e \in E(C)}(2hw(e)/\epsilon2^i + 1) = (2h/\epsilon 2^i)w(C) + |E(C)|$, which is at most $(2/\epsilon+1) h = h^*$. This implies that $M_i \le 2w_{G_i}(C)$, and so $M_s \le \epsilon2^iM_i/2h \le 2\cdot \epsilon2^iw_{G_i}(C)/2h \le 2(w(C) + \epsilon2^{i-1}) \le 2(1+\epsilon)w(C)$.

Now, we show that $w(C) \le M_l$ and, if $C$ contains $\ge h$ edges then $M_l \le (1+\epsilon)(2k-1) w(C)$. For proving the first condition, we have to show that $d'(s,v) + d'(s,x) + w(v,x) \ge w(C)$ holds for every edge $\{v,x\}$ that satisfies the condition in line \ref{line:condweight}. This is true because there always exists a cycle passing through $\{v,x\}$ whose weight does not exceed the computed estimate. Now, if $C$ contains at least $h$ edges, then with high probability, there is a vertex $s \in S \cap V(C)$. As $C$ is a cycle, there must be an edge $\{v,x\}$ that satisfies the condition in line \ref{line:condweight}. Thus, we have $M_l \le d'(s,v) + d'(s,x) + w(v,x) \le (1+\epsilon)(2k-1)(d(s,v) + d(s,x) +w(v,x)) \le (1+\epsilon)(2k-1)w(C)$. As $(2k-1)(1+\epsilon) \ge 2(1+ \epsilon)$, we conclude that $M \le (2k-1)(1+\epsilon)w(C)$.

\subparagraph{Complexity.} Step \ref{line:invoke} takes $\tilde{O}(h/\epsilon + k|S|^{1+1/k}+kD)$ rounds by Lemma \ref{lem:runEsti}. Step \ref{line:broa} takes $O(|S|)$ rounds. Step \ref{line:broaMs} takes $O(D)$ rounds. The loop at line \ref{line:loop} runs for $O(\log{n})$ iterations, each takes $\tilde{O}(n^{1/2} + h^* + D) = \tilde{O}(n^{1/2} +h + h/\epsilon + D)$ rounds by Lemma \ref{lem:bhopapp}. Thus, the total is $\tilde{O}(h+h/\epsilon + n^{1/2} + k|S|^{1+1/k}+kD)$ rounds. As, with high probability, we have $|S| = \tilde{O}(n/h) = \tilde{O}(n^{k/(2k+1)})$, we get that the running time is $\tilde{O}((k + 1/\epsilon) \cdot n^{(k+1)/(2k+1)}+kD)$.
\section{Directed Networks} \label{sec:dirunw}

We now focus on the directed setting and present an improved randomized algorithm for approximating the girth of unweighted, directed networks. Extension to weighted directed networks is discussed in a later subsection.

\subsection{Overview} 

We significantly improve the girth approximation framework of \cite{ManohR24} for approximating the girth in directed networks. As in previous work, cycles are divided into \emph{long}, which are likely to intersect a suitably sampled vertex set through which the cycles are computed exactly, and \emph{short}, which are handled via restricted BFS explorations of carefully defined local neighborhoods. Our main technical contribution is a global approach to handling congestion: instead of controlling congestion during the execution of restricted explorations (as was done in \cite{ManohR24}, and required an excessive number of rounds), we isolate and process highly congested vertices in advance, enabling a cleaner analysis and an improved round complexity of $\tilde{O}(n^{2/3}+D)$.

Throughout this section, let $G=(V,E)$ be a directed unweighted graph with $|V|=n$ and hop-diameter $D$. The algorithm maintains, for each vertex $v$, a value $\mu_v$
representing the length of the shortest directed cycle through $v$ discovered so far.
The final output is $\mu=\min_{v\in V}\mu_v$.

The algorithm is divided into two phases. The first phase is similar in spirit to the
corresponding phase of \cite{ManohR24}, and w.h.p.\ exactly detects the length
of any cycle with hop-length at least $h = n^{2/3}$. The second phase computes a
$2$-approximation for the remaining cycles, whose hop-length is at most $h$.

\subsection{Computing approximate for long MWC}

The algorithm sets $h = n^{2/3}$ and samples a set $S \subseteq V$ of size
$|S|=\tilde{\Theta}(n^{1/3})$ with high probability, by including each vertex independently with
probability $\tilde{\Theta}(1/h)$.

Next, the algorithm performs two multi-source exact directed BFS computations, one from
the vertices in $S$ and one to the vertices in $S$. As a result, each vertex $v \in V$
learns all distances $d(s , v)$ and $d(v , s)$ for every $s \in S$. This step takes
$O(\sqrt{n|S|}+D)=\tilde{O}(n^{2/3}+D)$ rounds with high probability \cite[Theorem~4]{ManohR24}. For every edge $(v,s)\in E$, the algorithm then locally updates
\[
\mu_v = \min\{ \mu_v,\; d(s,v) + 1 \}.
\]

Let $C$ be a minimum weight cycle in $G$. If the hop-length of $C$ is at least $h$, then
with high probability the sampled set satisfies $C \cap S \neq \emptyset$. In this case,
for every vertex $v \in C$ the algorithm has computed $d(v \lra s)$ for some
$s \in C \cap S$, and therefore sets $\mu_v = d(v \lra s)$.

Finally, for every pair $s,s' \in S$, the algorithm broadcasts the distance $d(s,s')$.
This step takes $O(|S|^2 + D)=\tilde{O}(n^{2/3})$ rounds and concludes the first phase.
The full pseudo-code is given in Algorithm~\ref{alg:unweighted-mwc}.

\linespace

\begin{algorithm}[H] \caption{2-Approximation Algorithm for Directed Unweighted MWC} \label{alg:unweighted-mwc} 
\KwIn{Directed unweighted graph $G=(V,E)$} \KwOut{$\mu$, a $2$-approximation of the minimum weight cycle in $G$} 
Let $h = n^{2/3}$\; 
Initialize $\mu_v \leftarrow \infty$ for all $v \in V$\; 
Sample a set $S \subseteq V$ by including each vertex independently with probability $\Theta\!\left(\frac{\log^3 n}{h}\right)$\; Compute distances $d(s,v)$ for all $s\in S$, $v\in V$ using multi-source directed BFS\; \tcc{Detect all cycles passing through sampled vertices} \ForEach{$s\in S$}{ \ForEach{incoming edge $(v,s)\in E$}{ $\mu_v \leftarrow \min\{\mu_v,\; d(s,v)+1\}$\; } } Broadcast all distances $\{d(s,t)\}_{s,t\in S}$ to all vertices\; Invoke Algorithm \ref{alg:second-phase}\; \tcc{the short-cycle procedure that processes cycles of hop-length $< h$.} Return $\mu \leftarrow \min_{v\in V}\mu_v$ computed by a convergecast. \end{algorithm}

\subsection{Computing approximate for short MWC}

The first phase "handles" long MWC, meaning that the MWC hop-set it at least $h = n^{2/3}$. In this section we present a new method for approximating the MWC if its hop distance is less than $h$. We assume that each vertex $v \in V$ knows all $d(v , s), d(s , v)$ for all $s \in S$, and all $d(s,s')$ for all $s,s'\in S$.

This subsection describes the construction of a vertex neighborhood $P(v)$ for each
$v\in V$, and a small \emph{elimination set} $R(v)\subseteq S$ that controls the size of
$P(v)$. Informally, the neighborhood $P(v)$ is defined so that (i) it contains all short cycles through $v$ whenever the minimum cycle through $v$ avoids the sampled set $S$ and is much smaller then the minimal $d(s \lra v)$ over all $s \in S$, and (ii) it is small enough to support a restricted exploration
in sublinear time.

\paragraph{Elimination predicate.}
We rely on the following elimination rule (cf.\ Lemma~5.1 in~\cite{CL21}).
For a fixed vertex $v$, a vertex $t\in S$ is said to \emph{eliminate} a vertex $y\in V$ if
\begin{equation}
\label{eq:elim}
2d(v,t) + d(t,y) \le 2d(v,y) + d(y,t).
\end{equation}
Lemma~5.1 in~\cite{CL21} states that whenever~\eqref{eq:elim} holds, the minimum cycle through $v$ and $t$ provides a $2$-approximation for the minimum cycle through $v$ and $y$, and since cycles through
sampled vertices are already handled in the first phase, the vertex $y$ can be excluded from the neighborhood of $v$.

\paragraph{Definition of $P(v)$ given $R(v)$.}
Given any subset $R(v)\subseteq S$, define
\begin{equation}
\label{eq:Pdef}
P(v) \;=\; \Bigl\{y\in V \,\Big|\, \forall t\in R(v),\;
2d(v,y) + d(y,t) \;<\; 2d(v,t) + d(t,y)\Bigr\}.
\end{equation}
Equivalently, $P(v)$ contains exactly those vertices not eliminated by any $t\in R(v)$.

\paragraph{Construction of $R(v)$.}
Fix $\beta=\lceil \log n\rceil$. Partition the sampled set $S$ into $\beta$ disjoint
groups $S_1,\ldots,S_\beta$ of (almost) equal size.\footnote{The partition can be fixed
globally using public randomness, or locally using a shared hash function.}
The set $R(v)$ is constructed iteratively, adding at most one representative from each
group. In iteration $i$, consider the subset of vertices in $S_i$ that do not eliminate
the current neighborhood, and select one such vertex (e.g.\ uniformly at random) to add
to $R(v)$. By construction, $|R(v)|\le \beta = O(\log n)$. Manoharan and Ramachandran show the exact algorithm to compute that directly, so we will later assume that we have $R(v)$ for every $v \in V$.

\begin{lemma}[Size bound]
\label{lem:P-size}
Assume $|S|=\Omega(\log n \cdot n^{1/3})$. Then for every fixed $v\in V$, the set $P(v)$
defined by~\eqref{eq:Pdef} using the above construction of $R(v)$ satisfies
$|P(v)| < n^{2/3}$ with high probability in $n$.
\end{lemma}

\begin{proof}[Proof sketch]
The proof follows the argument of Chechik and Lifshitz after substituting $h=n^{2/3}$ and $|S|=\Omega(\log n \cdot n^{1/3})$. The key point is that each iteration
adds one sampled vertex $t\in R(v)$ drawn from a fresh group $S_i$, and the elimination
predicate~\eqref{eq:elim} removes a constant fraction of the remaining vertices from the
candidate neighborhood (unless the remaining neighborhood is already of size at most
$\tilde{O}(n/|S|)$). After $\beta=\Theta(\log n)$ iterations, a standard multiplicative
shrinkage argument implies $|P(v)|=\tilde{O}(n/|S|)=\tilde{O}(n^{2/3})$ w.h.p.
\end{proof}

If the algorithm can compute the neighborhood $P(v)$ for any vertex $v\in V$ within
$\tilde{O}(n^{2/3}+D)$ rounds, then the remaining steps are straightforward. Once $P(v)$
is known, the algorithm can locally compute the minimum directed cycle through $v$
contained in $P(v)$, update the value $\mu_v$ accordingly, and broadcast this value to
the rest of the graph. A final convergecast then yields
$\mu=\min_{v\in V}\mu_v$, completing the computation.

It therefore suffices to show how to compute $P(v)$ efficiently. Since $|P(v)|<n^{2/3}$
with high probability (Lemma~\ref{lem:P-size}), computing $P(v)$ up to hop-distance $h=n^{2/3}$ can be done in $\tilde{O}(n^{2/3})$ rounds:

\paragraph{Restricted BFS.}
The key idea is that the algorithm computes $P(v)$ \emph{on the fly}, without explicitly constructing it in advance. The algorithm performs a BFS proceeding for exactly $h$ steps, and at each step a message is forwarded only to
neighbors that satisfy the membership condition of Definition~\ref{eq:Pdef}.

To test whether an intermediate vertex $y$ belongs to $P(v)$ before propagating the BFS, the algorithm uses the distances between sampled vertices that are already known from the
first phase, together with the information defining $R(v)$, which is carried in the BFS
message. Since $|R(v)|=O(\log n)$, the BFS message contains only $O(\log n)$ distance
values and therefore fits within the CONGEST bandwidth constraints.

As a result, the restricted BFS explores exactly the vertices of $P(v)$ up to hop-distance $h$, and completes in $\tilde{O}(n^{2/3})$ rounds. The problem becomes when trying to compute all $P(v)$ simultaneously. 

\paragraph{Congestion measure.}
For a subset $U \subseteq V$ and a vertex $u \in V$, define
\[
P^{-1}(u, U) \;=\; \{\, v \in U \mid u \in P(v) \,\},
\qquad
p^{-1}(U) \;=\; \max_{u \in V} |P^{-1}(u,U)|.
\]
Intuitively, $P^{-1}(u,U)$ measures how many neighborhoods $P(v)$, for $v\in U$, contain
the vertex $u$, and $p^{-1}(U)$ captures the maximum congestion induced by the family
$\{P(v)\}_{v\in U}$.

\paragraph{Computing neighborhoods under bounded congestion.}
Manoharan and Ramachandran showed that, by introducing random delays and carefully
scheduling restricted BFS explorations, one can compute all neighborhoods $P(v)$,
$v \in U$, in time that depends linearly on the maximum congestion parameter
$p^{-1}(U)$, up to polylogarithmic factors. Adapting their scheduling technique to our
setting yields the following guarantee.

\begin{lemma}
\label{lem:restricted-bfs-scheduling}
Let $U \subseteq V$ and let $h = n^{2/3}$. Then the restricted $h$-hop BFS explorations
rooted at all vertices $v \in U$ can be executed in
\[
\tilde{O}\bigl(p^{-1}(U) + h \bigr)
\]
rounds in the CONGEST model.
\end{lemma}

\begin{proof}[Proof sketch.]
Each vertex $v\in U$ performs an $h$-hop BFS restricted to $P(v)$, as described above.
To control congestion, the algorithm assigns each such BFS a random start delay drawn
from a suitable range, ensuring that BFS messages originating from different sources are
sufficiently spread out over time. With high probability, no vertex $u$ is required to
process messages from more than $\tilde{O}(p^{-1}(U))$ sources simultaneously. Since each restricted BFS proceeds for
$h=n^{2/3}$ steps, the total running time is
$\tilde{O}(p^{-1}(U) + h)$.
\end{proof}

For convenience, we refer to this procedure as
$\mathsf{SchedRBFS}(U,h)$, and use it as a black box in the remainder of the algorithm, invoking Lemma~\ref{lem:restricted-bfs-scheduling}
whenever its running time is needed.

\paragraph{Second phase: intuition.}
The algorithm does not need to compute the neighborhoods $P(v)$ for all vertices
simultaneously. In fact, it suffices to compute $P(v)$ for a single vertex $v$ on any
directed cycle $C$ whose weight is less than half of the smallest value $\mu_v$ found so
far. Thus, the goal of the second phase is to ensure that at least one vertex on every
such short cycle can be processed.

If the congestion parameter satisfies $p^{-1}(V)=O(n^{2/3})$, then the algorithm can invoke the procedure $\mathsf{SchedRBFS}(V,h)$, and compute all neighborhoods $P(v)$ and terminate. The main challenge arises from vertices that belong to too many neighborhoods $P(v)$, which would cause excessive congestion.

The key idea is therefore to identify and process highly congested vertices first,
remove them from the graph, and then recurse on the remaining vertices.

\paragraph{Handling congested vertices.} We define, for each $U \subseteq V$ and $u \in U$,
\[
p^{-1}(u,U) := |P^{-1}(u,U)|.
\]
A vertex $u$ is called \emph{overflowing} with respect to $U$ if
\[
p^{-1}(u,U) > 2n^{2/3}.
\]

Let $B_0 := V$. Assume first that no vertex is overflowing with respect to $B_0$. Then $p^{-1}(B_0)=O(n^{2/3})$, and the algorithm can invoke $\mathsf{SchedRBFS}(B_0,h)$ and compute 
Otherwise, define
\[
B_1 = \{\, u \in V : p^{-1}(u,B_0) > 2n^{2/3} \,\}
\]
to be the set of overflowing vertices with respect to $B_0$. Our goal is to process the vertices in $B_1$ first and then remove them from the graph. The algorithm could do that if $p^{-1}(B_1) = O(n^{2/3})$ - but that could not be the case also. For $i\ge 1$, given $B_i$, let
\[
B_{i+1} = \{\, u \in V : p^{-1}(u,B_i) > 2n^{2/3} \,\}.
\]
This process continues until $|B_\tau| < 2n^{2/3}$, because after that $B_{\tau+1} = \emptyset$.

\paragraph{Shrinking of congestion layers.}
The following lemma shows that this peeling process terminates quickly.

\begin{lemma}
For all $i \ge 0$: (1) $B_{i+1} \subseteq B_i$ and (2) $|B_{i+1}| \le \frac{1}{2}|B_i|$.

In particular, the process terminates after at most $O(\log n)$ iterations.
\end{lemma}

\begin{proof}[Proof sketch.]
(1) If $u\in B_{i+1}$, then $u$ belongs to $P(v)$ for more than $2n^{2/3}$ vertices $v\in B_i$. Since $B_i \subseteq B_{i-1}$ by induction, this implies $u\in B_i$.

(2) Each vertex $v \in B_i$ can contribute to at most $|P(v)| \le n^{2/3}$ vertices $u$ for
which $v \in P^{-1}(u,B_i)$. Consequently,
\[
\sum_{u \in V} p^{-1}(u,B_i) \;\le\; |B_i|\cdot n^{2/3}.
\]
On the other hand, by definition, every vertex $u \in B_{i+1}$ satisfies
$p^{-1}(u,B_i) > 2n^{2/3}$. Therefore,
\[
\sum_{u \in V} p^{-1}(u,B_i)
\;\ge\;
\sum_{u \in B_{i+1}} p^{-1}(u,B_i)
\;>\;
|B_{i+1}|\cdot 2n^{2/3}.
\]
Combining the two bounds yields
\[
|B_{i+1}|\cdot 2n^{2/3} < |B_i|\cdot n^{2/3},
\]
and hence $|B_{i+1}| \le \tfrac12 |B_i|$.
\end{proof}

Ideally, the algorithm would explicitly compute the layers
$B_0,\ldots,B_\tau$, and once $|B_\tau|<2n^{2/3}$, compute all neighborhoods $P(v)$ for
$v\in B_\tau$, update the values $\mu_v$, remove $B_\tau$, and continue with
$B_{\tau-1},B_{\tau-2},\ldots$. After removing $B_{i+1}$, there are no overflowing vertices anymore with respect to $B_i$, so computing $\mathsf{SchedRBFS}(B_i,h)$ on $G \setminus B_{i+1}$ takes $\tilde{O}(h)$ time.

The difficulty is that the algorithm cannot afford to compute these layers exactly. Determining whether a vertex $u$ belongs to $B_1$ requires knowing the exact value of
$p^{-1}(u,V)$, which in turn depends on the neighborhoods $P(v)$ for many sources $v$. Instead, the algorithm relies on approximate congestion information: given an approximation $\tilde{B}_i$ of $B_i$, we show that by sampling $\tilde{\Theta}(n^{1/3})$ vertices from $\tilde{B}_i$, the algorithm can construct a good approximation of $B_{i+1}$.

\subsection{Second Phase: Algorithm Overview}
\label{subsec:phase2-overview}

The second phase approximates the minimum weight cycle among cycles of hop-length at most
$h = n^{2/3}$ that avoid the sampled set $S$. It consists of two stages: Constructing a sequence of approximate congestion layers
    \(
    V = \tilde{B}_0 \supseteq \tilde{B}_1 \supseteq \tilde{B}_2 \supseteq \cdots \supseteq \tilde{B}_\tau
    \), and processing these layers from $\tilde{B}_\tau$ up to $\tilde{B}_0$, computing restricted neighborhoods and updating the values $\mu_v$.
    
\paragraph{Constructing the approximate layers
$\tilde{B}_1,\ldots,\tilde{B}_\tau$.}
The construction proceeds iteratively. Set $\tilde{B}_0 := V$.
At iteration $i \ge 0$, the algorithm first computes $|\tilde{B}_{i}|$ via a standard convergecast on a BFS tree in $O(D)$ rounds and checks whether $|\tilde{B}_{i}| > 0$. If so, the construction terminates and we set
$\tau := i-1$.

Otherwise, the algorithm samples a set
$S_{i} \subseteq \tilde{B}_{i}$ by including each vertex independently with probability
$p: =\frac{c \log n}{h}$ for a large enough constant $c$, and invokes $\mathsf{SchedRBFS}(S_{i},h)$. As a result, every vertex
$u \in \tilde{B}_{i}$ can compute the exact value $p^{-1}(u,S_{i})$.

Using this information, the algorithm defines the next approximate layer as
\[
\tilde{B}_{i+1} \;=\; \{\, u \in \tilde{B}_{i} \mid p^{-1}(u,S_{i}) > 3c \log n \,\},
\]
and proceeds to the next iteration.

\paragraph{Processing the layers.}
Once the sequence $\tilde{B}_0,\tilde{B}_1,\ldots,\tilde{B}_\tau$ has been constructed, the
algorithm processes the layers in reverse order.
Starting from $\tilde{B}_\tau$ and continuing up to $\tilde{B}_0$, the algorithm invokes
$\mathsf{SchedRBFS}(\tilde{B}_i \setminus \tilde{B}_{i+1},h)$, computes the restricted neighborhoods $P(v)$ for all
$v \in \tilde{B}_i \setminus \tilde{B}_{i+1}$, and updates the values $\mu_v$ accordingly.
After processing a layer $\tilde{B}_i$, its vertices are excluded from all subsequent
calls to $\mathsf{SchedRBFS}$ (conceptually removing them from the graph).

\begin{algorithm}[H]
\caption{\textsc{SecondPhase}: Handling short cycles via approximate congestion layers}
\label{alg:second-phase}
\KwIn{Directed graph $G=(V,E)$; hop bound $h=n^{2/3}$; access to membership test for $P(v)$ via $R(v)$; black box $\mathsf{SchedRBFS}(U,h)$.}
\KwOut{Updated values $\mu_v$ for all $v\in V$ (hence $\mu=\min_v \mu_v$ after convergecast).}

\BlankLine

\tcc{\bf Stage 1: Construct approximate congestion layers}
$\tilde{B}_0 \leftarrow V, i \leftarrow 0$\;

\While{$|\tilde{B}_{i}| \ge 0$}{
    \tcc{Compute $|\tilde{B}_{i}|$ using standard methods in $O(D)$ time}
    Sample a set $S_{i} \subseteq \tilde{B}_{i}$ by including each vertex with probability $\frac{c\log n}{h}$\;
    Invoke $\mathsf{SchedRBFS}(S_{i},h)$\;
    \tcc{Each vertex $u$ can now compute $p^{-1}(u,S_{i})$ exactly}
    $\tilde{B}_{i+1} \leftarrow \{\, u \in V \mid p^{-1}(u,S_{i}) > 3c \log n \,\}$\;
    $i \leftarrow i+1$\;
}
$\tau \leftarrow i-1$\;

\BlankLine
\tcc{\bf Stage 2: Process layers in reverse order}
\For{$i \leftarrow \tau$ \KwTo $0$}{
    $U_i \leftarrow \tilde{B}_i \setminus \tilde{B}_{i+1}$\;
    \tcc{Compute restricted neighborhoods for sources in the current layer}
    $\mathsf{SchedRBFS}(U_i,h)$\;
    Update $\mu_v$ for every $v \in U_i$\;
    Exclude $U_i$ from subsequent iterations (conceptually remove it from $G$).
}
\end{algorithm}

\subsection{Analysis}
We now present details of the proof of correctness and round complexity of Algorithm \ref{alg:unweighted-mwc}.

\begin{lemma}
Algorithm~\ref{alg:unweighted-mwc} computes a $2$-approximation of the minimum weight
cycle in a directed unweighted graph $G=(V,E)$ with high probability in $n$.
\end{lemma}

\begin{proof}
Let $C$ be a minimum weight cycle in $G$, and let $g:=|C|$ denote its hop-length.

If the hop-length of $C$ is at least $h=n^{2/3}$, then with high probability the sampled
set $S$ intersects $C$ in the first phase. In this case, the algorithm computes the
exact length of $C$ by combining distances from and to the sampled vertex, and therefore
returns the correct value.

We henceforth assume that the hop-length of $C$ is less than $h$.
If after the first phase there exists a vertex $v$ with $\mu_v \le 2g$, then the
algorithm has already found a valid $2$-approximation and we are done.

Otherwise, consider the execution of the second phase. Since $C$ does not intersect
$S$, Lemma~5.1 of~\cite{CL21} implies that for every vertex $v\in C$, the entire cycle
$C$ is contained in the neighborhood $P(v)$.

Let $v\in C$ be the first vertex on $C$ that is processed by the algorithm during the
layered peeling procedure of the second phase. At the time $v$ is processed, all other
vertices of $C$ have not yet been removed, and therefore the restricted BFS rooted at
$v$ explores the entire cycle $C$ within $h$ hops. Consequently, the algorithm
discovers the exact length of $C$ during this exploration and sets $\mu_v = g$.
\end{proof}

We now bound the round complexity by showing that the approximate layers shrink
geometrically.

\begin{lemma}
\label{lem:detected-correct}
Fix an iteration $i$ and a vertex $u\in \tilde B_i$.
\begin{enumerate}
    \item If $p^{-1}(u,\tilde B_i)\ge 4h$, then
    \[
    \p[u\notin \tilde B_{i+1}] \le n^{-\Theta(c)}.
    \]
    \item If $p^{-1}(u,\tilde B_i)\le 2h$, then
    \[
    \p[u\in \tilde B_{i+1}] \le n^{-\Theta(c)}.
    \]
\end{enumerate}
\end{lemma}

\begin{proof}
Recall that $\tilde B_{i+1}$ consists of all vertices $u$ satisfying
$p^{-1}(u,S_i) \ge 3c\log n$. Note that
\[
p^{-1}(u, S_i) \sim \mathrm{Bin}\!\left(p^{-1}(u,\tilde B_i),\frac{c\log n}{h} \right),
\]
and let $\mu := \mathbb{E}[p^{-1}(u, S_i)]
= \frac{c\log n}{h} \cdot p^{-1}(u,\tilde B_i)$.

\smallskip
\noindent
(1) If $p^{-1}(u,\tilde B_i) \ge 4h$, then $\mu \ge 4c\log n$. By a Chernoff bound,
\[
\p[u \notin \tilde B_{i+1}]
=
\p\big[p^{-1}(u,S_i) < 3c\log n\big]
=
\p\big[p^{-1}(u,S_i) < \tfrac{3}{4}\mu\big]
\le n^{-\Theta(c)}.
\]

\smallskip
\noindent
(2) If $p^{-1}(u,\tilde B_i) \le 2h$, then $\mu \le 2c\log n$. By a Chernoff bound,
\[
\p[u \in \tilde B_{i+1}]
=
\p\big[p^{-1}(u,S_i) \ge 3c\log n\big]
\le n^{-\Theta(c)}.
\]
\end{proof}

\begin{lemma}
\label{lem:layer-shrink-correct}
Assume $|P(v)|\le h$ for all $v\in \tilde B_i$.
with high probability in $n$, for all iterations
$i\le \lceil \log n\rceil$,
\[
|\tilde B_{i+1}| \le \tfrac12 |\tilde B_i|.
\]
Consequently, $\tau=O(\log n)$ with high probability.
\end{lemma}

\begin{proof}
Fix an iteration $i$.
By double counting,
\[
\sum_{u\in V} p^{-1}(u,\tilde B_i)
=
\sum_{v\in \tilde B_i} |P(v)|
\le
|\tilde B_i|\cdot h.
\]
Hence, the number of vertices $u$ satisfying $p^{-1}(u,\tilde B_i) > 2h$
is at most $|\tilde B_i|/2$.

By Lemma~\ref{lem:detected-correct}(2), each vertex $u$ with
$p^{-1}(u,\tilde B_i)\le 2h$ is included in $\tilde B_{i+1}$
with probability at most $n^{-\Theta(c)}$.
Taking a union bound over all vertices $u\in V$, with probability at least
$1-n^{-\Theta(c)+1}$, no such vertex is included in $\tilde B_{i+1}$.

Condition on this event. Then every vertex in $\tilde B_{i+1}$ satisfies
$p^{-1}(u,\tilde B_i) > 2h$, and therefore
\[
|\tilde B_{i+1}| \le \tfrac12 |\tilde B_i|.
\]

Finally, taking a union bound over all iterations
$i\le \lceil \log n\rceil$ and choosing $c$ sufficiently large completes the proof.
\end{proof}

\begin{lemma}[Round complexity of the second phase]
\label{lem:phase2-rounds}
Algorithm~\ref{alg:second-phase} completes in $\tilde{O}(h+D)=\tilde{O}(n^{2/3}+D)$
rounds with high probability in $n$.
\end{lemma}

\begin{proof}
\noindent\textbf{Stage 1 (constructing layers).}
In iteration $i$, the algorithm first computes $|\tilde B_i|$ by a standard aggregation
(convergecast and broadcast) on a BFS tree, which costs $O(D)$ rounds.

Next, it samples $S_i$ by including each vertex with probability
$p=\frac{c\log n}{h}$ and invokes procedure $\mathsf{SchedRBFS}(S_i,h)$.
By Lemma~\ref{lem:restricted-bfs-scheduling}, this costs
\[
\tilde{O}\bigl(p^{-1}(S_i) + h \bigr)
\]
rounds. 

By Lemma~\ref{lem:layer-shrink-correct}, with high probability in $n$ we have
$|\tilde B_{i+1}| \le \tfrac12|\tilde B_i|$ for all $i\le \lceil\log n\rceil$, and hence
the number of iterations in Stage~1 is $\tau=O(\log n)$.

Moreover, for every $i\le \tau$ the sampling is performed from $\tilde B_i$ (each vertex
outside $\tilde B_i$ is irrelevant), so
\[
\mathbb{E}[|S_i|] = p\cdot |\tilde B_i| \le \frac{c\log n}{h}\cdot n = O(n^{1/3}\log n),
\]
and a Chernoff bound implies $|S_i|=\tilde{O}(n^{1/3})$ for all $i\le \tau$ with high
probability. Therefore,
\[
p^{-1}(S_i)\le |S_i|=\tilde{O}(n^{1/3})
\]
deterministically, and Lemma~\ref{lem:restricted-bfs-scheduling} gives that each call
$\mathsf{SchedRBFS}(S_i,h)$ runs in $\tilde{O}(h)$ rounds w.h.p.

Summing over $\tau=O(\log n)$ iterations, Stage~1 takes
\[
\tilde{O}(h) + O(D)\cdot O(\log n) = \tilde{O}(h + D)
\]
rounds with high probability.

\noindent\textbf{Stage 2 (processing layers).}
For each $i=\tau,\ldots,0$, the algorithm invokes $\mathsf{SchedRBFS}(U_i,h)$ where
$U_i=\tilde B_i\setminus \tilde B_{i+1}$, and the sets $U_i$ form a partition of $V$.

We claim that for every $i$, $p^{-1}(U_i)=O(h)$ with high probability.
Fix $i$ and a vertex $u\in U_i$. Since $u\in \tilde B_i$ but $u\notin \tilde B_{i+1}$,
Lemma~\ref{lem:detected-correct}(1) implies that, except with probability $n^{-\Theta(c)}$,
\[
p^{-1}(u, U_i) < 4h .
\]
Applying Lemma~\ref{lem:restricted-bfs-scheduling} gives that each call
$\mathsf{SchedRBFS}(U_i,h)$ completes in $\tilde{O}(h)$ rounds.
Taking a union bound over all $i \le \tau$ and choosing $c$ sufficiently large, all these bounds hold simultaneously with high probability. Hence Stage~2 takes $\tilde{O}(h)$ rounds in total.
\end{proof}

\subsection{Extension to Weighted Networks}
We use the same transformation as Manoharan and Ramachandran \cite{ManohR24}. We run the unweighted algorithm on weighted networks by simulating each directed edge of weight $w$ as an unweighted directed path of length $w$; that is, each BFS step of the original algorithm is simulated by a sequence of steps along a virtual path. To prevent the running time from getting too big, we limit all BFS executions on the virtual network to a given depth $h^*$. This approach guarantees a $2$-approximation whenever the girth has weight $O(h^*)$, and achieves a running time of $\tilde{O}(n^{2/3} + h^* + D)$. 

Next, we apply the same scaling that we used in Algorithm \ref{alg:weiundire}. Namely, we construct $O(\log{h^*W})$ versions of $G$. In the $i$-th version $G_i$, the weight of each edge $e$ is set to $\lceil 2h^*w(e)/\eps 2^i\rceil$. By an identical reasoning to that in Algorithm \ref{alg:weiundire}, if $G$ contains a minimum-weight cycle $C$ with at most $h^*$ edges, then there exists an index $i$ such that $w_{G_i}(C) = O(h^*)$, and $(\eps2^i/2h^*) \cdot w_{G_i}(C) \le (1+\eps)w(C)$. Therefore, applying the above procedure to all graphs and taking the minimum yields a $(2+\eps)$-approximation for a minimum-weight cycle with hop bound $h^*$ in $\tilde{O}(n^{2/3}  + h^* + D)$ rounds. In particular, for $h^* = \Theta(n^{2/3})$, the running time becomes $\tilde{O}(n^{2/3} + D)$.

To handle longer cycles, we sample a set $S$ of $\tilde{\Theta}(n^{1/3})$ vertices by selecting each vertex independently with probability $\tilde{\Theta}(n^{-2/3})$. We then compute $(1+\eps)$-approximate distances $\tilde{d}$ from all $S$ in time $\tilde{O}(\sqrt{|S|n}+D) = \tilde{O}(n^{2/3} + D)$ \cite[Theorem 4]{ManohR24}, and compute the estimation $\min_{v}\min\{\tilde{d}(s,v) + w(s,v) \mid s \in S$ is an ingoing neighbor of $v\}$ by convergecasting in $O(D)$ rounds. If there exists a minimum-weight cycle $C$ of length at least $n^{2/3}$, then with high probability at least one of its vertices is sampled, and so we get a $(1+\eps)$-approximation.

\section{Lower Bounds}
\label{sec:lower-bounds}

In this section we present new lower bounds for distributed girth approximation in the
\CONGEST\ model. Our main lower bound concerns the directed, unweighted setting and shows that
approximating the girth within any factor strictly smaller than~$3$ requires $\tilde{\Omega}(n^{3/5})$ rounds, even on graphs of (undirected) diameter $O(\log n)$.
We later extend this result to any integer approximation factors~$\alpha$, as well as to the undirected weighted setting. 

Before presenting the lower bound constructions, we recall a standard notions from extremal
graph theory. Let $\mathcal{F}$ be a family of graphs. A graph is called $\mathcal{F}$-free if it contains no copy of a graph in $\mathcal{F}$
as a subgraph. A cornerstone of extremal graph theory is the study of the \textit{Turan number} $ex(n,\mathcal{F})$, the maximum number of edges in an $\mathcal{F}$-free graph on $n$ vertices. Define the \textit{Zarankiewicz number} $\z(n,\mathcal{F})$ to be the maximum number of edges in an $\mathcal{F}$-free bipartite graph on $n$ vertices. Let $C_k$ denote a cycle of length $k$, and let $C_{<k}$ denote the set of cycles
$C_\ell$, where $3 \le \ell < k$. 

It is easy to see that $\ex(n ,C_{< 2k}) = \Theta \left( \z(n,C_{<2k}) \right)$, and determining their values has been extensively studied.
Bondy and Simonovits~\cite{BS1974} showed that $\ex (n,C_{< 2k}) = O(n^{1+1/k})$ for every $k$. It is widely conjectured that this bound is tight, namely that
$\ex(n,C_{< 2k}) = \Theta(n^{1+1/k})$, see~\cite{ErdosSimonovits1982}.
This conjecture is currently known to hold only for $k \in \{3,4,6\}$~\cite{Brown1966,FUREDINAOR2006}.

\paragraph{Two-Party Communication.}
Our lower bounds are obtained via reductions from two-party communication problems, most
notably \emph{Set Disjointness}, as well as from graph problems with known \CONGEST\ lower bounds.
In the Set Disjointness problem, two players, Alice and Bob, are each given an $m$-bit string
$E_A$ and $E_B$, respectively. Their goal is to determine whether the two sets are disjoint,
that is, whether there exists no index $i \in \{1,\ldots,m\}$ such that $E_A[i] = E_B[i] = 1$.

A classical result in communication complexity states that any (possibly randomized) protocol
for Set Disjointness requires the exchange of $\Omega(m)$ bits, even when shared randomness is
allowed~\cite{KalyanasundaramS92, Razborov92}. Consequently, any \CONGEST\ lower bound obtained via a reduction from Set Disjointness applies equally to randomized distributed algorithms.

\paragraph{High-Level Intuition.}
We describe a reduction showing that if there exists an algorithm in the \CONGEST\ model
for directed graphs on $n$ vertices that computes a $(3-\eps)$-approximation of the girth
in $R$ rounds, then one can solve Set Disjointness on $n^{3/5}$ bits within $\tilde O(R)$ rounds as well.
Consequently, it must hold that $R = \tilde \Omega(n^{3/5})$.

Let $H = (L \cup R, E(H))$ be a bipartite graph on $n^{2/5}$ vertices, $|L|=|R|=\frac{1}{2}n^{2/5}$, with girth at least $6$ and
\[
\z(n^{2/5}, C_{<6}) = \Theta \left((n^{2/5})^{3/2} \right) = \Theta(n^{3/5})
\]
edges.
Alice and Bob each receive subsets of its edges, denoted
$E_A, E_B \subseteq E(H)$.
They must determine whether $E_A \cap E_B = \emptyset$.

Roughly, the reduction constructs a directed graph $G$ as follows.
First, we create two disjoint copies of the graph defined by the edge sets
$E_A$ and $E_B$, on vertex sets $L_A \cup R_A$ and $L_B \cup R_B$, respectively.
The edges in Alice's copy are directed from $L_A$ to $R_A$,
whereas the edges in Bob's copy are directed from $R_B$ to $L_B$.

Next, for every vertex $r_A \in R_A$ and its corresponding copy $r_B \in R_B$,
we connect $r_A$ to $r_B$ by a directed path (a ``pipe'') of length $n^{3/5} - 1$,
directed from Alice's side to Bob's side.
Symmetrically, for every pair of corresponding vertices
$l_A \in L_A$ and $l_B \in L_B$,
we add a directed pipe of length $n^{3/5} - 1$ directed from $l_B$ to $l_A$.

Thus, the only way to move from Alice's part to Bob's part (or vice versa)
is by traversing an entire pipe.
Since $|L \cup R| = n^{2/5}$ and each such vertex contributes a pipe
of length $n^{3/5}$, the total number of vertices in $G$
is $n^{2/5} \cdot n^{3/5} = n$.
We now consider two cases.

If Alice and Bob share an edge $(l,r)$, then there is a directed cycle in $G$ that uses this single shared edge:
the cycle goes from $l_A$ to $r_A$ (Alice's edge), traverses the pipe from $r_A$ to $r_B$,
continues along $(r_B,l_B)$ (Bob's edge), and returns to $l_A$ via the pipe from $l_B$ to $l_A$.
Since each pipe has length $n^{3/5}-1$, this cycle has length
\[
2\cdot (n^{3/5}-1) + 2 = 2n^{3/5}.
\]

If the inputs are disjoint, any directed cycle must alternate between Alice's and Bob's parts,
and every transition between the two parts requires traversing an entire pipe.
Contracting each pipe to a single edge yields a cycle in the underlying bipartite graph $H$.
Since $H$ has girth at least $6$, any such cycle uses at least six extremal edges,
and therefore traverses at least six pipes.
Hence any directed cycle in this case has length at least
\[
6\cdot (n^{3/5}-1) + 6 = 6n^{3/5}.
\]

We thus obtain a multiplicative gap: in the intersecting case the girth is $2n^{3/5}$,
whereas in the disjoint case it is at least $6n^{3/5}$.
Any $(3-\varepsilon)$-approximation algorithm for girth must distinguish between these two possibilities.
We will show that Alice and Bob can simulate such a \CONGEST\ algorithm, and thus obtain a
two-party protocol that solves \textsc{Set Disjointness}.

\subsubsection{The Lower Bound Construction}

Let $a \ge 1$ be an integer, and let $H=(L \cup R, E(H))$ be a bipartite graph
with $|L|=|R|=a$ that contains no cycles of length $<6$ and maximizes the number of edges.
Denote
\[
m := |E(H)| = \z(a,C_{<6}) =\Theta(a^{3/2}).
\]
Fix an arbitrary ordering $E(H)=\{e_1,\ldots,e_m\}$, where each
$e_i=(l,r)$ satisfies $l\in L$ and $r\in R$.

We consider an instance of the \textsc{Set Disjointness} problem in which Alice and Bob are
given bit vectors $E_A,E_B\in\{0,1\}^m$, representing subsets of $E(H)$.
The goal is to determine whether $E_A$ and $E_B$ are disjoint.

From this instance we construct a directed, unweighted graph $G$. See figure~\ref{fig:lb-graph}. Let \(q := \frac{n}{2a}.\)

\begin{figure}[t]
\centering
\begin{tikzpicture}[
    x=1.2cm, y=0.8cm,
    every node/.style={circle, draw, minimum size=6mm, inner sep=0pt, font=\small}
]
\begin{scope}[rotate=-90]

\node (r1a) at (0, 3.3) {$r_1$};
\node (v10) at (1, 3.3) {$r_2$};
\node (v11) at (2, 3.3) {$r_3$};
\node (v12) at (3, 3.3) {$r_4$};
\node[draw=none] (v1d) at (4, 3.3) {$\vdots$};
\node (r1b) at (5, 3.3) {$r_q$};

\node (r2a) at (-0.5, 2.2) {$r'_1$};
\node (v20) at (0.5, 2.2) {$r'_2$};
\node (v21) at (1.5, 2.2) {$r'_3$};
\node (v22) at (2.5, 2.2) {$r'_4$};
\node[draw=none] (v2d) at (3.5, 2.2) {$\vdots$};
\node (r2b) at (4.5, 2.2) {$r'_q$};

\node (r3a) at (-1, 1.1) {$r^*_1$};
\node (v30) at (0, 1.1) {$r^*_2$};
\node (v31) at (1, 1.1) {$r^*_3$};
\node (v32) at (2, 1.1) {$r^*_4$};
\node[draw=none] (v3d) at (3, 1.1) {$\vdots$};
\node (r3b) at (4, 1.1) {$r^*_q$};

\node[draw=none] at (-0.8, 0.6) {$\ddots$};
\node[draw=none] at (0.2, 0.6) {$\ddots$};
\node[draw=none] at (1.25, 0.6) {$\ddots$};
\node[draw=none] at (2.6, 0.6) {$\ddots$};

\node (raa) at (-1.5, 0) {$r`_1$};
\node (va0) at (-0.5, 0) {$r`_2$};
\node (va1) at (0.5, 0) {$r`_3$};
\node (va2) at (1.5, 0) {$r`_4$};
\node[draw=none] (vad) at (2.5, 0) {$\vdots$};
\node (rab) at (3.5, 0) {$r`_q$};


\draw[->] (r1a) -- (v10);
\draw[->] (v10) -- (v11);
\draw[->] (v11) -- (v12);
\draw[->] (v12) -- (v1d);
\draw[->] (v1d) -- (r1b);

\draw[->] (r2a) -- (v20);
\draw[->] (v20) -- (v21);
\draw[->] (v21) -- (v22);
\draw[->] (v22) -- (v2d);
\draw[->] (v2d) -- (r2b);

\draw[->] (r3a) -- (v30);
\draw[->] (v30) -- (v31);
\draw[->] (v31) -- (v32);
\draw[->] (v32) -- (v3d);
\draw[->] (v3d) -- (r3b);

\draw[->] (raa) -- (va0);
\draw[->] (va0) -- (va1);
\draw[->] (va1) -- (va2);
\draw[->] (va2) -- (vad);
\draw[->] (vad) -- (rab);




\node (l1a) at (0, -3.2) {$l_1$};
\node (u12) at (1, -3.2) {$l_2$};
\node (u13) at (2, -3.2) {$l_3$};
\node (u14) at (3, -3.2) {$l_4$};
\node[draw=none] (u1d) at (4, -3.2) {$\vdots$};
\node (l1b) at (5, -3.2) {$l_q$};

\node (l2a) at (-0.5, -4.3) {$l'_1$};
\node (u22) at (0.5, -4.3) {$l'_2$};
\node (u23) at (1.5, -4.3) {$l'_3$};
\node (u24) at (2.5, -4.3) {$l'_4$};
\node[draw=none] (u2d) at (3.5, -4.3) {$\vdots$};
\node (l2b) at (4.5, -4.3) {$l'_q$};

\node (l3a) at (-1, -5.4) {$l^*_1$};
\node (u32) at (0, -5.4) {$l^*_2$};
\node (u33) at (1, -5.4) {$l^*_3$};
\node (u34) at (2, -5.4) {$l^*_4$};
\node[draw=none] (u3d) at (3, -5.4) {$\vdots$};
\node (l3b) at (4, -5.4) {$l^*_q$};

\node[draw=none] at (-0.8, -5.9) {$\ddots$};
\node[draw=none] at (0.2, -5.9) {$\ddots$};
\node[draw=none] at (1.25, -5.9) {$\ddots$};
\node[draw=none] at (2.6, -5.9) {$\ddots$};

\node (laa) at (-1.5, -6.5) {$l`_1$};
\node (ua2) at (-0.5, -6.5) {$l`_2$};
\node (ua3) at (0.5, -6.5) {$l`_3$};
\node (ua4) at (1.5, -6.5) {$l`_4$};
\node[draw=none] (uad) at (2.5, -6.5) {$\vdots$};
\node (lab) at (3.5, -6.5) {$l`_q$};


\draw[->] (l1b) -- (u1d);
\draw[->] (u1d) -- (u14);
\draw[->] (u14) -- (u13);
\draw[->] (u12) -- (l1a);
\draw[->] (u13) -- (u12);

\draw[->] (l2b) -- (u2d);
\draw[->] (u2d) -- (u24);
\draw[->] (u24) -- (u23);
\draw[->] (u23) -- (u22);
\draw[->] (u22) -- (l2a);

\draw[->] (l3b) -- (u3d);
\draw[->] (u3d) -- (u34);
\draw[->] (u34) -- (u33);
\draw[->] (u33) -- (u32);
\draw[->] (u32) -- (l3a);

\draw[->] (lab) -- (uad);
\draw[->] (uad) -- (ua4);
\draw[->] (ua4) -- (ua3);
\draw[->] (ua3) -- (ua2);
\draw[->] (ua2) -- (laa);

\draw[->, purple] (l1a) -- (r2a);
\draw[->, dashed, gray] (l1a) -- (r3a);

\draw[->, dashed, gray] (l2a) -- (r1a);
\draw[->, purple] (l2a) -- (raa);

\draw[->,dashed, gray] (l3a) -- (r1a);
\draw[->,purple] (l3a) -- (raa);

\draw[->,dashed, gray] (laa) -- (r2a);
\draw[->,purple] (laa) -- (r3a);

\draw[->,purple] (r1b) -- (l2b);
\draw[->,dashed, gray] (r1b) -- (l3b);

\draw[->,purple] (r2b) -- (l1b);
\draw[->,dashed, gray] (r2b) -- (lab);

\draw[->,dashed, gray] (r3b) -- (l1b);
\draw[->,dashed, gray] (r3b) -- (lab);

\draw[->,dashed, gray] (rab) -- (l2b);
\draw[->,purple] (rab) -- (l3b);

\node[draw=none, circle=none] at (-1.8,-3.25) {Alice's Side $L_1\cup R_1$};
\node[draw=none, circle=none] at (5.4,0.05) {Bob's Side $L_q \cup R_q$};
\end{scope}

\end{tikzpicture}

\vspace{-42pt}

\caption{The graph construction $G$ (Shortcut Tree is omitted). Here $r,r',r^*,r` \in R$ and $l,l',l^*,l` \in L$. The dashed gray directed edges correspond to the set of all potential edges (i.e., if $E_A = E_B = E(H)$), while the pink edges correspond to the actual 
edge sets given to Alice and Bob.
}
\label{fig:lb-graph}
\end{figure}
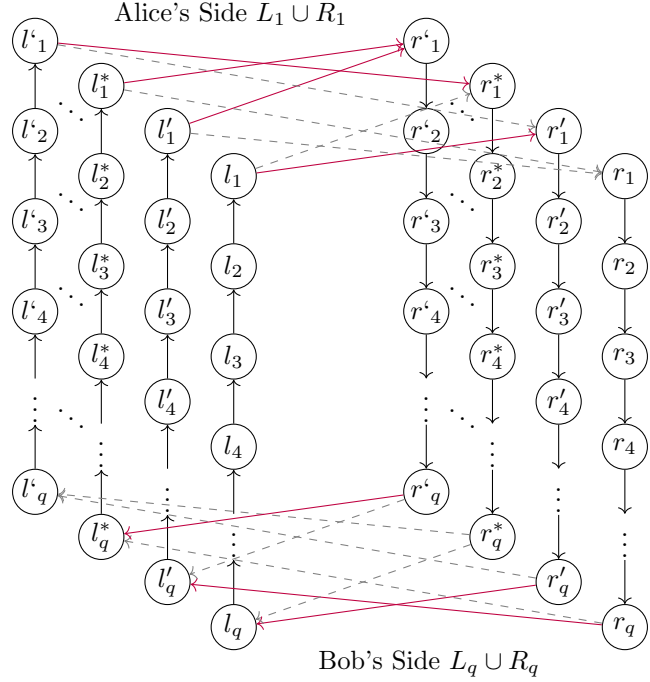

\paragraph{Vertex copies (layers).}
We create $q$ copies of the vertex set $L \cup R$.
For every $i \in \{1,\dots,q\}$, let $L_i$ and $R_i$
denote copies of $L$ and $R$, respectively.
For $l\in L$, denote its copy in layer $i$ by $l_i$,
and similarly for $r\in R$.

\paragraph{Pipes (input independent).}
For every $r\in R$, add the directed path
\[
r_1 \to r_2 \to \cdots \to r_q.
\]

For every $l\in L$, add the directed path
\[
l_q \to l_{q-1} \to \cdots \to l_1.
\]

\paragraph{Input-dependent edges.}
For each $i\in[m]$, let $e_i=(l,r)\in E(H)$.
If $E_A[i]=1$, add the directed edge $(l_1,r_1)$.
If $E_B[i]=1$, add the directed edge $(r_q,l_q)$.

\paragraph{Shortcut Tree - enforcing small undirected diameter.}
So far, the construction may have large undirected diameter (due to the length-$q$ pipes).
To ensure that the resulting instance also has undirected diameter $O(\log n)$, we add an
auxiliary directed tree gadget that is \emph{one-way} and therefore cannot participate in any
directed cycle.

We add a balanced binary tree $T$ with $q$ leaves $t_1,\ldots,t_q$ and height $O(\log q)$,
and orient all tree edges toward the leaves.
For every layer $i\in[q]$, we add directed edges from \emph{every} vertex in $L_i\cup R_i$
to the corresponding leaf $t_i$. We call $T$ the \emph{Shortcut Tree} of $G$.

Since all edges incident to $T$ are directed \emph{into} $T$ and all edges of $T$ point toward
the leaves, no directed cycle can use a vertex of $T$. On the other hand, in the underlying undirected graph, every vertex has an $O(\log q)$-length
path to the root of $T$, thus the undirected diameter is $O(\log q)$.

We start by formalizing the key structural property of the construction.

\begin{claim}
\label{clm:girth-gap}
Let $q$ be the pipe parameter.
\begin{enumerate}
    \item If $E_A \cap E_B \neq \emptyset$, then $\girth(G) = 2q$.
    \item If $E_A \cap E_B = \emptyset$, then $\girth(G) \ge 6q$.
\end{enumerate}
\end{claim}

\begin{proof}
We first observe two structural properties of $G$.

(1) No directed cycle contains a vertex of the shortcut tree.

(2) Any directed cycle must alternate between Alice's and Bob's parts.
The only edges connecting the two parts are full pipe paths,
and input-dependent edges appear only in layer $1$ (Alice)
or layer $q$ (Bob).
Thus every transition between the two parts requires traversing an entire pipe.

\paragraph{Intersecting inputs.}
If $E_A \cap E_B \neq \emptyset$, let $(l,r)\in E(H)$ belong to both.
Then
\[
l_1 \to r_1 \to \cdots \to r_q \to l_q \to \cdots \to l_1
\]
is a directed cycle of length $2q$.
By Property~(2), any directed cycle must traverse at least two pipes,
and therefore has length at least $2q$.
Hence $\girth(G)=2q$.

\paragraph{Disjoint inputs.}
Assume $E_A \cap E_B = \emptyset$, and let $C$ be a directed cycle.
By (1), $C$ avoids the tree.
By (2), $C$ alternates between the two parts,
and each alternation requires a full pipe traversal.

Projecting the input-dependent edges of $C$ onto $H$
yields a cycle in $H$.
Since $H$ has girth at least $6$,
$C$ must contain at least six input-dependent edges,
and therefore traverse at least six pipes.
Each alternation contributes exactly $q$ edges,
so $|C| \ge 6q$.
Thus $\girth(G)\ge 6q$.
\end{proof}

\paragraph{A black-box communication lower bound.}
To formalize the reduction, we invoke a result of Das Sarma \etal (Lemma 4.1 in~\cite{DasSarmaEtAl2011}).
They proved that on the underlying undirected communication network consisting of
the pipes and the shortcut tree (i.e., the input-independent edges of our construction),
any randomized \CONGEST\ algorithm that solves \textsc{Set Disjointness}
with constant error probability requires
\(
\tilde{\Omega}\!\left(\min(q,m)\right)
\)
rounds. We now reduce \textsc{Set Disjointness} to directed girth approximation.

\begin{theorem}
\label{thm:girth-lb}
Let $G$ be the directed unweighted graph constructed above from inputs
$E_A,E_B\in\{0,1\}^m$.
Any randomized $(3-\varepsilon)$-approximation algorithm for the directed girth
on $G$ (for any $\varepsilon>0$)
requires \( \tilde{\Omega}\!\left(\min(q,m)\right) \)
rounds.
\end{theorem}

\begin{proof}
Alice and Bob locally construct their respective parts of $G$
according to $E_A$ and $E_B$,
and simulate the girth-approximation algorithm.

By Claim~\ref{clm:girth-gap},
if $E_A \cap E_B \neq \emptyset$ then $\girth(G)=2q$,
whereas if $E_A \cap E_B = \emptyset$ then $\girth(G)\ge 6q$.
Since
\[
(3-\varepsilon)\cdot 2q < 6q,
\]
any $(3-\varepsilon)$-approximation distinguishes the two cases,
and hence decides whether $E_A \cap E_B \neq \emptyset$.
Thus the algorithm yields a randomized protocol for
\textsc{Set Disjointness} with the same round complexity. 

By the black-box lower bound applied to the underlying undirected communication network,
any such protocol requires
\(
\tilde{\Omega}\!\left(\min(q,m)\right)
\)
rounds.
\end{proof}

We have $q = \frac{n}{2a}$. Substituting into $\tilde{\Omega}(\min\{q,m\})$ and using
$m=\Theta(a^{3/2})$ gives
\[
\tilde{\Omega}\!\left(
\min\!\left\{
\frac{n}{a},\;
a^{3/2}
\right\}
\right).
\]
Balancing $\frac{n}{a} = a^{3/2}$ yields $a = n^{2/5}$,
and hence the lower bound is $\tilde{\Omega}(n^{3/5})$.

\subsection{The Undirected Weighted Construction}

We now extend the construction to the undirected weighted setting.
The vertex set and edge set remain exactly as in the directed case,
except that all edges are made undirected and assigned weights.
In this setting, $\girth(G)$ denotes the minimum total weight of a cycle.

Let $\varepsilon>0$. We assign weights as follows:

\begin{itemize}
    \item Every pipe edge has weight $1$.
    \item Every input-dependent edge has weight $\frac{q}{\varepsilon}$.
    \item Every shortcut-tree edge (including pipe-to-leaf edges)
    has weight $\frac{6q}{\varepsilon}$.
\end{itemize}

\begin{claim}
\label{clm:weighted-gap}
\begin{enumerate}
    \item If $E_A \cap E_B \neq \emptyset$, then
    \[
        \girth(G) = \frac{2q}{\varepsilon}(1+\varepsilon).
    \]
    \item If $E_A \cap E_B = \emptyset$, then
    \[
        \girth(G) \ge \frac{6q}{\varepsilon}.
    \]
\end{enumerate}
\end{claim}

\begin{proof}
\textbf{Intersecting inputs.}
As in the directed case, if $(l,r)\in E_A\cap E_B$,
there exists a cycle using two pipes and two input-dependent edges.
Its total weight is
\[
2q \cdot 1 + 2\cdot \frac{q}{\varepsilon}
= \frac{2q}{\varepsilon}(1+\varepsilon).
\]
\medskip

\textbf{Disjoint inputs.}
Assume $E_A \cap E_B=\emptyset$, and let $C$ be a minimum-weight cycle.

If $C$ uses any shortcut-tree edge,
then $w(C)\ge \frac{6q}{\varepsilon}$ and we are done.

Otherwise, $C$ consists only of pipe and input-dependent edges.
As in the directed case, projecting $C$ onto $H$
yields a cycle in $H$.
Since $H$ has girth at least $6$,
$C$ must contain at least six input-dependent edges.
Therefore,
\[
w(C) \ge 6\cdot \frac{q}{\varepsilon}
= \frac{6q}{\varepsilon}.
\]

\end{proof}

The reduction to a $(3-3\varepsilon)$-approximation algorithm follows exactly the same argument as in the directed unweighted case, and we therefore only sketch it.

Assume there exists a $(3-3\varepsilon)$-approximation algorithm for weighted girth.
By Claim~\ref{clm:weighted-gap},
in the \textsc{Yes}-instance
\[
\girth(G)=\frac{2q}{\varepsilon}(1+\varepsilon),
\]
while in the \textsc{No}-instance
\[
\girth(G)\ge \frac{6q}{\varepsilon}.
\]

Since
\[
(3-3\varepsilon)\cdot \frac{2q}{\varepsilon}(1+\varepsilon)
= \frac{6q}{\varepsilon}(1-\varepsilon^2)
< \frac{6q}{\varepsilon},
\]
such an algorithm distinguishes the two cases,
and the communication lower bound follows as in the directed case.

\subsection{General girth approximation}

The previous construction extends naturally to larger approximation factors.
Let $k \ge 2$, and let
$H_{2k}=(L\cup R,E_{2k})$ be a bipartite graph with $|L|=|R|=a$
that contains no cycle of length less than $2k$ and maximizes the number of edges.
Assuming the Erd\H{o}s--Simonovits even cycle conjecture,
\[
m := |E_{2k}| = \Theta(a^{1+1/(k-1)}).
\]

We construct the graph $G$ exactly as before, replacing $H$ with $H_{2k}$.

\paragraph{Cycle structure.}
In the \textsc{Yes}-instance there exists a directed cycle traversing two pipes,
of length $2q$.
In the \textsc{No}-instance, any directed cycle projects to a cycle in $H_{2k}$,
and therefore must contain at least $2k$ input-dependent edges.
Hence it traverses at least $2k$ pipes and has length at least $2kq$.

Thus distinguishing the two cases requires a $(k-\varepsilon)$-approximation.

\paragraph{Optimizing the parameters.}
The graph contains
\(
n = 2aq
\)
vertices, and the lower bound is $\tilde{\Omega}(\min(q,m))$.
Balancing $q=m$ and using $m=\Theta(a^{1+1/(k-1)})$ gives
\[
\frac{n}{a} = a^{1+1/(k-1)}
\]
Solving yields
\[
a = \Theta\!\left(n^{\frac{k-1}{2k-1}}\right),
\qquad
q = \Theta\!\left(n^{\frac{k}{2k-1}}\right)
\]

We obtain the following theorem.

\begin{theorem}
Let $k \ge 3$ and $\varepsilon>0$.
Assuming the Erd\H{o}s--Simonovits even cycle conjecture (or unconditionally for $k\in\{3,4,6\}$),
any $(k-\varepsilon)$-approximation algorithm for girth in the
CONGEST model requires
\[
\tilde{\Omega}\!\left(n^{\frac{k}{2k-1}}\right)
\]
rounds, even when the (undirected) diameter is $O(\log n)$.
This holds for directed graphs and for undirected weighted graphs.
\end{theorem}

Even without assuming the Erd\H{o}s--Simonovits conjecture, we can plug in the best
known explicit constructions of $C_{2k}$-free bipartite graphs due to
Lazebnik--Ustimenko--Woldar~\cite{LUW94}.
These yield $\z(a,C_{2k})=\Omega(a^{1+\gamma_k})$ for an explicit $\gamma_k > \frac{2}{3k}$.
Substituting this bound into our reduction gives an unconditional round lower bound
$\tilde\Omega(n^{1/2 + \beta_k})$ for an explicit exponent $\beta_k > \frac{1}{6k}$, thus strictly
improving over the previously known $\tilde\Omega(\sqrt{n})$ bound.

\bibliographystyle{plain}
\bibliography{ref}

\end{document}